\documentclass{article}
\usepackage[margin=1in]{geometry}
\usepackage[utf8]{inputenc}
\usepackage[dvipsnames]{xcolor}
\usepackage{hyperref,url}
\hypersetup{
    colorlinks,
    linkcolor=black,
    citecolor=blue,
    urlcolor={blue!80!black}
}
\usepackage[T1]{fontenc}
\usepackage{physics}
\usepackage{tikz}
\usepackage{graphicx,subcaption}
\usepackage{float}
\usepackage{authblk}

\usepackage{csquotes}
\usepackage[normalem]{ulem}

\usepackage{amsmath,amsfonts,amsthm,amssymb,dsfont,bbm,bm}
\theoremstyle{plain}
\newtheorem{proposition}{Proposition}
\newtheorem{corollary}{Corollary}

\newcommand{\bbE}{\mathbb{E}}
\newcommand{\bbR}{\mathbb{R}}
\newcommand{\bbN}{\mathbb{N}}
\newcommand{\bbZ}{\mathbb{Z}}
\newcommand{\rmd}{\mathrm{d}}
\newcommand{\rmS}{\mathrm{S}}
\newcommand{\ES}{\mathrm{ES}}
\newcommand{\VS}{\mathrm{VS}}
\newcommand{\QS}{\mathrm{QS}}
\newcommand{\BS}{\mathrm{BS}}
\newcommand{\CRPS}{\mathrm{CRPS}}
\newcommand{\SE}{\mathrm{SE}}
\newcommand{\rmAE}{\mathrm{AE}}
\newcommand{\DSS}{\mathrm{DSS}}
\newcommand{\ESS}{\mathrm{ESS}}
\newcommand{\logS}{\mathrm{LogS}}
\newcommand{\AS}{\mathrm{AS}}

\newcommand{\bmx}{\bm{x}}
\newcommand{\bmX}{\bm{X}}
\newcommand{\bmY}{\bm{Y}}
\newcommand{\bmy}{\bm{y}}
\newcommand{\bms}{\bm{s}}
\newcommand{\calP}{\mathcal{P}}
\newcommand{\calD}{\mathcal{D}}
\newcommand{\calS}{\mathcal{S}}
\newcommand{\calT}{\mathcal{T}}
\newcommand{\calF}{\mathcal{F}}
\newcommand{\calL}{\mathcal{L}}
\newcommand{\calN}{\mathcal{N}}

\usepackage[round]{natbib}
\AtBeginDocument{\def\doi#1{\url{https://doi.org/#1}}}

\title{Proper Scoring Rules for Multivariate Probabilistic Forecasts based on Aggregation and Transformation}
\author[1]{Romain Pic}
\author[1]{Clément Dombry}
\author[2]{Philippe Naveau}
\author[3]{Maxime Taillardat}

\affil[1]{\small Université de Franche Comté, CNRS, LmB (UMR 6623), F-25000 Besançon, France}
\affil[2]{Laboratoire des Sciences du Climat et de l'Environnement, UMR 8212, CEA-CNRS-UVSQ, EstimR, IPSL \& U Paris-Saclay, Gif-sur-Yvette, France}
\affil[3]{CNRM, Université de Toulouse, Météo-France, CNRS, Toulouse, France}

\begin{document}
	
\maketitle
\begin{abstract}
    Proper scoring rules are an essential tool to assess the predictive performance of probabilistic forecasts. However, propriety alone does not ensure an informative characterization of predictive performance and it is recommended to compare forecasts using multiple scoring rules. With that in mind, interpretable scoring rules providing complementary information are necessary. We formalize a framework based on aggregation and transformation to build interpretable multivariate proper scoring rules. Aggregation-and-transformation-based scoring rules are able to target specific features of the probabilistic forecasts; which improves the characterization of the predictive performance. This framework is illustrated through examples taken from the literature and studied using numerical experiments showcasing its benefits. In particular, it is shown that it can help bridge the gap between proper scoring rules and spatial verification tools.
\end{abstract}

\section{Introduction}

Probabilistic forecasting allows to issue forecasts carrying information about the prediction uncertainty. It has become an essential tool in numerous applied fields such as weather and climate prediction \citep{Vannitsem2021, Palmer2012}, earthquake forecasting \citep{Jordan2011, Schorlemmer2018}, electricity price forecasting \citep{Nowotarski2018} or renewable energies \citep{Pinson2013WindEnergy, Gneiting2023} among others. Moreover, it is slowly reaching fields further from "usual" forecasting, such as epidemiology predictions \citep{Bosse2023} or breast cancer recurrence prediction \citep{AlMasry2023}. In weather forecasting, probabilistic forecasts often take the form of ensemble forecasts in which the dispersion among members captures forecast uncertainty.\\

The development of probabilistic forecasts has induced the need for appropriate verification methods. Forecast verification fulfills two main purposes: quantifying how good a forecast is given observations available and allowing one to rank different forecasts according to their predictive performance. Scoring rules provide a single value to compare forecasts with observations. Propriety is a property of scoring rules that encourages forecasters to follow their true beliefs and that prevents hedging. Proper scoring rules allow to assess calibration and sharpness simultaneously \citep{Winkler1977, Winkler1996}. Calibration is the statistical compatibility between forecasts and observations. Sharpness is the uncertainty of the forecast itself. Propriety is a necessary property of good scoring rules, but it does not guarantee that a scoring rule provides an informative characterization of predictive performance. In univariate and multivariate settings, numerous studies have proven that no scoring rule has it all, and thus, different scoring rules should be used to get a better understanding of the predictive performance of forecasts (see, e.g., \citealt{Scheuerer2015, Taillardat2021, Bjerregaard2021}). With that in mind, \cite{Scheuerer2015} "strongly recommend that several different scores be always considered before drawing conclusions." This amplifies the need for numerous complementary proper scoring rules that are well-understood to facilitate forecast verification. In that direction, \cite{Dorninger2018} states that: "gaining an in-depth understanding of forecast performance depends on grasping the full meaning of the verification results." Interpretability of proper scoring rules can arise from being induced by a consistent scoring function for a functional (e.g., the squared error is induced by a scoring function consistent for the mean; \citealt{Gneiting2011}), knowing what aspects of the forecast the scoring rule discriminates (e.g., the Dawid-Sebastiani score only discriminates forecasts through their mean and variance; \citealt{Dawid1999}) or knowing the limitations of a certain proper scoring rule (e.g., the variogram score is incapable of discriminating two forecasts that only differ by a constant bias; \citealt{Scheuerer2015}). In this context, interpretable proper scoring rules become verification methods of choice as the ranking of forecasts they produce can be more informative than the ranking of a more complex but less interpretable scoring rule. Section~\ref{section:verification-review} provides an in-depth explanation of this in the case of univariate scoring rules. It is worth noting that interpretability of a scoring rule can also arise from its decomposition into meaningful terms (see, e.g., \citealt{Broecker2009}). This type of interpretability can be used complementarily to the framework proposed in this article.\\ 

\cite{Scheuerer2015} proposed the variogram score to target the verification of the dependence structure. The variogram score of order $p$ ($p>0$) is defined as
\[
    \VS_p(F,\bmy)=\sum_{i,j=1}^{d} w_{ij}\left(\bbE_F\left[|X_i-X_j|^p\right]-|y_i-y_j|^p\right)^2,
\]
where $X_i$ is the $i$-th component of the random vector $\bmX\in\bbR^d$ following $F$, the $w_{ij}$ are nonnegative weights and $\bmy\in\bbR^d$ is an observation. The construction of the variogram score relies on two main principles. First, the variogram score is the weighted sum of scoring rules acting on the distribution of $\bmX_{i,j}=(X_i,X_j)$ and on paired components of the observations $y_{i,j}$. This \textit{aggregation} principle allows the combination of proper scoring rules and summarizes them into a proper scoring rule acting on the whole distribution $F$ and observations $\bmy$. Second, the scoring rules composing the weighted sum can be seen as a standard proper scoring rule applied to transformations of both forecasts and observations. Let us denote $\gamma_{i,j}: \bmx \mapsto |x_{i}-x_{j}|^p$ the transformation related to the variogram of order $p$, then the variogram score can be rewritten as 
\[
    \VS_p(F,\bmy) = \sum_{i,j=1}^d w_{ij} \SE(\gamma_{i,j}(F),\gamma_{i,j}(\bmy)),
\]
where $\SE(F,y)=(\bbE_F[X]-y)^2$ is the univariate squared error and $\gamma_{i,j}(F)$ is the distribution of $\gamma_{i,j}(\bmX)$ for $\bmX$ following $F$. This second principle is the \textit{transformation} principle, allowing to build transformation-based proper scoring rules that can benefit from interpretability arising from a transformation (here, the variogram transformation $\gamma_{i,j}$) and the simplicity and interoperability of the proper scoring rule they rely on (here, the squared error). 

We review the univariate and multivariate proper scoring rules through the lens of interpretability and by mentioning their known benefits and limitations. We formalize these two principles of aggregation and transformation to construct interpretable proper scoring rules for multivariate forecasts. To illustrate the use of these principles, we provide examples of transformation-and-aggregation-based scoring rules from both the literature on probabilistic forecast verification and quantities of interest. We conduct a simulation study to empirically demonstrate how transformation-and-aggregation-based scoring rules can be used. Additionally, we show how the aggregation and transformation principle can help bridging the gap between the proper scoring rules framework and the spatial verification tools \citep{Gilleland2009, Dorninger2018}.\\

The remainder of this article is organized as follows. Section~\ref{section:verification-review} gives a general review of verification methods for univariate and multivariate forecasts. Section~\ref{section:framework} introduces the framework of proper scoring rules based on transformation and aggregation for multivariate forecasts. Section~\ref{section:example} provides examples of transformation-and-aggregation-based scoring rules, including examples from the literature. Then, Section~\ref{section:sim-study} showcases through different simulation setups how the framework proposed in this article can help build interpretable proper scoring rules. Finally, Section~\ref{section:conclusion} provides a summary as well as a discussion on the verification of multivariate forecasts. Throughout the article, we focus on spatial forecasts for simplicity. However, the points made remain valid for any multivariate forecasts, including temporal forecasts or spatio-temporal forecasts.

\section{Overview of verification tools for univariate and multivariate forecasts}\label{section:verification-review}

This section presents the zoology of available verification tools and briefly summarizes their benefits and limitations. First, we define scoring rules and their key properties. Then, we recall univariate scoring rules, starting with ones derived from scoring functions used in point forecasting. Finally, we provide an overview of verification tools for multivariate forecasts.

\subsection{Calibration, sharpness, and propriety}

\cite{Gneiting2007Paradigm} proposed a paradigm for the evaluation of probabilistic forecasts: "maximizing the sharpness of the predictive distributions subject to calibration".  \textit{Calibration} is the statistical compatibility between the forecast and the observations. \textit{Sharpness} is the concentration of the forecast and is a property of the forecast itself. In other words, the paradigm aims at minimizing the uncertainty of the forecast given that the forecast is statistically consistent with the observations. \cite{Tsyplakov2011} states that the notion of calibration in the paradigm is too vague but it holds if the definition of calibration is refined. This principle for the evaluation of probabilistic forecasts has reached a consensus in the field of probabilistic forecasting (see, e.g., \citealt{Gneiting2014, Thorarinsdottir2018}). The paradigm proposed in \cite{Gneiting2007Paradigm} is not the first mention of the link between sharpness and calibration: for example, \cite{Murphy1987} mentioned the relation between refinement (i.e., sharpness) and calibration.\\

For univariate forecasts, multiple definitions of calibration are available depending on the setting. The most used definition is \textit{probabilistic calibration} and, broadly speaking, consists of computing the rank of observations among samples of the forecast and checking for uniformity with respect to observations. If the forecast is calibrated, observations should not be distinguishable from forecast samples, and thus, the distribution of their ranks should be uniform. Probabilistic calibration can be assessed by probability integral transform (PIT) histograms \citep{Dawid1984} or rank histograms \citep{Anderson1996, Talagrand1997} for ensemble forecasts when observations are stationary (i.e., their distribution is the same across time). The shape of the PIT or rank histogram gives information about the type of (potential) miscalibration: a triangular-shaped histogram suggests that the probabilistic forecast has a systematic bias, a $\cup$-shaped histogram suggests that the probabilistic forecast is under-dispersed and a $\cap$-shaped histogram suggests that the probabilistic forecast is over-dispersed. Moreover, probabilistic calibration implies that rank histograms should be uniform but uniformity is not sufficient. For example, rank histograms should also be uniform conditionally on different forecast scenarios (e.g., conditionally on the value of the observations available when the forecast is issued). Additionally, under certain hypotheses, calibration tools have been developed to consider real-world limitations such as serial dependence \citep{Broecker2020}. Statistical tests have been developed to check the uniformity of rank histograms \citep{Jolliffe2008}. Readers interested in a more in-depth understanding of univariate forecast calibration are encouraged to consult \cite{Tsyplakov2013, Tsyplakov2020}.

For multivariate forecasts, a popular approach relies on a similar principle: first, multivariate forecast samples are transformed into univariate quantities using so-called pre-rank functions and then the calibration is assessed by techniques used in the univariate case (see, e.g., \citealt{Gneiting2008}). Pre-rank functions may be interpretable and allow targeting the calibration of specific aspects of the forecast such as the dependence structure. Readers interested in the calibration of multivariate forecasts can refer to \cite{Allen2024Assessing} for a comprehensive review of multivariate calibration.\\

A scoring rule $\rmS$ assigns a real-valued quantity $\rmS(F,y)$ to a forecast-observation pair $(F,y)$, where $F\in\calF$ is a probabilistic forecast and $\bmy\in\bbR^d$ is an observation. In the negative-oriented convention, a scoring rule $\rmS$ is \textit{proper relative to the class} $\calF$ if 
\begin{equation}\label{eq:proper}
    \bbE_G[\rmS(G,\bmY)]\leq\bbE_G[\rmS(F,\bmY)]
\end{equation}
for all $F,G\in\calF$, where $\bbE_G[\cdots]$ is the expectation with respect to $\bmY\sim G$. In simple terms, a scoring rule is proper relative to a class of distribution if its expected value is minimal when the true distribution is predicted, for any distribution within the class. Forecasts minimizing the expected scoring rule are said to be \textit{efficient} and the other forecasts are said to be \textit{sub-efficient}. Moreover, the scoring rule $\rmS$ is \textit{strictly proper relative to the class} $\calF$ if the equality in \eqref{eq:proper} holds if and only if $F=G$. This ensures the characterization of the ideal forecast (i.e., there is a unique efficient forecast and it is the true distribution).  Moreover, proper scoring rules are powerful tools as they allow the assessment of calibration and sharpness simultaneously \citep{Winkler1977, Winkler1996}. Sharpness can be assessed individually using the entropy associated with proper scoring rules, defined by $e_\rmS(F)=\bbE_F[\rmS(F,\bmY)]$. The sharper the forecast, the smaller its entropy. Strictly proper scoring rules can also be used to infer the parameters of a parametric probabilistic forecast (see, e.g., \citealt{Gneiting2005, Pacchiardi2024}).\\

Regardless of all the interesting properties of proper scoring rules, it is worth noting that they have some limitations. Proper scoring rules may have multiple efficient forecasts (i.e., associated with their minimal expected value) and, in the general setting, no guarantee is given on their relevance. Moreover, strict propriety ensures that the efficient forecast is unique and that it is the ideal forecast (i.e., the true distribution), however, no guarantee is available for forecasts within the vicinity of the minimum in the general case. This is particularly problematic since, in practice, the unavailability of the ideal distribution makes it impossible to know if the minimum expected score is achieved. In the case of calibrated forecasts, the expected scoring rule is the entropy of the forecast and the ranking of forecasts is thus linked to the information carried by the forecast (see Corollary 4, \citealt{Holzmann2014} for the complete result). These limitations may explain the plurality of scoring rules depending on application fields.\\

\subsection{Univariate scoring rules}

We recall classical univariate scoring rules to explain key concepts. Some univariate scoring rules will be useful for the multivariate scoring rules construction framework proposed in Section~\ref{section:framework}. Let $\calP(E)$ denote the class of Borel probability measures on $E$. We consider $F\in\calF\subseteq\calP(\bbR)$ a probabilistic forecast in the form of its cumulative distribution function (cdf) and $y\in\bbR$ an observation. When the probabilistic forecast $F$ has a probability density function (pdf), it will be denoted $f$.  

The simplest scoring rules can be derived from scoring functions used to assess point forecasts. The squared error (SE) is the most popular and is known through its averaged value (the mean squared error; MSE) or the square root of its average (the root mean squared error; RMSE) which has the advantage of being expressed in the same units as the observations. As a scoring rule, the SE is expressed as
\begin{equation}\label{eq:SE}
    \SE(F,y) = (\mu_F-y)^2,
\end{equation}
where $\mu_F$ denotes the mean of the predicted distribution $F$. The SE solely discriminates the mean of the forecast (see Appendix~\ref{appendix:expected_scores}); efficient forecasts for SE are the ones matching the mean of the true distribution. The SE is proper relative to $\calP_2(\bbR)$, the class of Borel probability measures on $\bbR$ with a finite second moment (i.e., finite variance). Note that the SE cannot be strictly proper as the equality of mean does not imply the equality of distributions. 

Another well-known scoring rule is the absolute error (AE) defined by
\begin{equation}\label{eq:AE}
    \rmAE(F,y) = |\mathrm{med}(F)-y|,
\end{equation}
where $\mathrm{med}(F)$ is the median of the predicted distribution $F$. The mean absolute error (MAE), the average of the absolute error, is the most seen form of the AE and it is also expressed in the same units as the observations. Efficient forecasts are forecasts that have a median equal to the median of the true distribution. The AE is proper relative to $\calP_1(\bbR)$ but not strictly proper. Similarly, the quantile score (QS), also known as the pinball loss, is a scoring rule focusing on quantiles of level $\alpha$ defined by 
\begin{equation}\label{eq:QS}
    \QS_\alpha(F,y) = (\mathds{1}_{y\leq F^{-1}(\alpha)}-\alpha)(F^{-1}(\alpha)-y)
\end{equation}
where $0<\alpha<1$ is a probability level and $F^{-1}(\alpha)$ is the predicted quantile of level $\alpha$. The case $\alpha=0.5$ corresponds to the AE up to a factor $2$. The QS of level $\alpha$ is proper relative to $\calP_1(\bbR)$ but not strictly proper since efficient forecasts are ones correctly predicting the quantile of level $\alpha$ (see, e.g., \citealt{Friederichs2008}).

Another summary statistic of interest is the exceedance of a threshold $t\in\bbR$. The Brier score (BS; \citealt{Brier1950}) was initially introduced for binary predictions but allows also to discriminate forecasts based on the exceedance of a threshold $t$. For probabilistic forecasts, the BS is defined as
\begin{equation}\label{eq:BS}
    \BS_t(F,y) = ((1-F(t))-\mathds{1}_{y>t})^2= (F(t)-\mathds{1}_{y\leq t})^2,
\end{equation}
where $1-F(t)$ is the predicted probability that the threshold $t$ is exceeded. The BS is proper relative to $\calP(\bbR)$ but not strictly proper. Binary events (e.g., exceedance of thresholds) are relevant in weather forecasting as they are used, for example, in operational settings for decision-making.

All the scoring rules presented above are proper but not strictly proper since they only discriminate against specific summary statistics instead of the whole distribution. Nonetheless, they are still used as they allow forecasters to verify specific characteristics of the forecast: the mean, the median, the quantile of level $\alpha$ or the exceedance of a threshold $t$. The simplicity of these scoring rules makes them interpretable, thus making them essential verification tools.\\

Some univariate scoring rules contain a summary statistic: for example, the formulas of the QS \eqref{eq:QS} or the BS \eqref{eq:BS} contain the exceedance of a threshold $t$ and the quantile of level $\alpha$, respectively. They can be seen as a scoring function applied to a summary statistic. This duality can be understood through the link between scoring functions and scoring rules through consistent functionals as presented in \cite{Gneiting2011} or Section~2.2 in \cite{Lerch2017}.

Other summary statistics can be of interest depending on applications. Nonetheless, it is worth noting that mispecifications of numerous summary statistics cannot be discriminated because of their \textit{non-elicitability}. Non-elicitability of a transformation implies that no proper scoring rule can be constructed such that efficient forecasts are forecasts where the transformation is equal to the one of the true distribution. For example, the variance is known to be non-elicitable; however, it is jointly elicitable with the mean (see, e.g., \citealt{Brehmer2017}). Readers interested in details regarding elicitable, non-elicitable and jointly elicitable transformations may refer to \cite{Gneiting2011}, \cite{Brehmer2019} and references therein.\\

A strictly proper scoring rule should discriminate the whole distribution and not only specific summary statistics. The continuous ranked probability score (CRPS; \citealt{Matheson1976}) is the most popular univariate scoring rule in weather forecasting applications and can be expressed by the following expressions
\begin{align}
    \CRPS(F,y) &= \bbE_F|X-y|-\frac{1}{2}\bbE_F|X-X'|,\label{eq:CRPS}\\ 
               &= \int_\bbR \BS_z(F,y) \rmd z,\label{eq:CRPS_bs}\\ 
               &= 2 \int_0^1 \QS_\alpha(F,y) \rmd\alpha,\label{eq:CRPS_qs}
\end{align}
where $y\in\bbR$ and $X$ and $X'$ are independent random variables following $F$, with a finite first moment. Equations~\eqref{eq:CRPS_bs} and \eqref{eq:CRPS_qs} show that the CRPS is linked with the BS and the QS. Broadly speaking, as the QS discriminates a quantile associated with a specific level, integrating the QS across all levels discriminates the quantile function that fully characterizes univariate distributions. Similarly, integrating the BS across all thresholds discriminates the cumulative distribution function that also fully characterizes univariate distributions. The CRPS is a strictly proper scoring rule relative to $\calP_1(\bbR)$, the class of Borel probability measures on $\bbR$ with a finite first moment. In addition, Equation~\eqref{eq:CRPS} indicates the CRPS values have the same units as observations. In the case of deterministic forecasts, the CRPS reduces to the absolute error, in its scoring function form \citep{Hersbach2000}. The use of the CRPS for ensemble forecast is straightforward using expectations as in \eqref{eq:CRPS}. \cite{Ferro2008} and \cite{Zamo2017} studied estimators of the CRPS for ensemble forecasts.\\

In addition to scoring rules based on scoring functions, some scoring rules use the moments of the probabilistic forecast $F$. The SE \eqref{eq:SE} depends on the forecast only through its mean $\mu_F$. The Dawid-Sebastiani score (DSS; \citealt{Dawid1999}) is a scoring rule depending on the forecast $F$ only through its first two central moments. The DSS is expressed as
\begin{equation}\label{eq:dss_univariate}
    \DSS(F,y) = 2\log(\sigma_F) + \frac{(\mu_F-y)^2}{{\sigma_F}^2},
\end{equation}
where $\mu_F$ and ${\sigma_F}^2$ are the mean and the variance of the distribution $F$. The DSS is proper relative to $\calP_2(\bbR)$ but not strictly proper, since efficient forecasts only need to correctly predict the first two central moments (see Appendix~\ref{appendix:expected_scores}). \cite{Dawid1999} proposed a more general class of proper scoring rules but the DSS, as defined in~\eqref{eq:dss_univariate}, can be seen as a special case of the logarithmic score (up to an additive constant), introduced further down.

Another scoring rule relying on the central moments of the probabilistic forecast $F$ up to order three is the error-spread score (ESS; \citealt{Christensen2014}). The ESS is defined as
\begin{equation}\label{eq:ESS}
    \ESS(F,y) = ({\sigma_F}^2-(\mu_F-y)^2-(\mu_F-y)\sigma_F\gamma_F)^2,
\end{equation}
where $\mu_F$, $\sigma_F^2$ and $\gamma_F$ are the mean, the variance and the skewness of the probabilistic forecast $F$. The ESS is proper relative to $\calP_4(\bbR)$. As for the other scoring rules only based on moments of the forecast presented above, the expected ESS compares the probabilistic forecast $F$ with the true distribution only via their four first moments (see Appendix~\ref{appendix:expected_scores}). Scoring rules based on central moments of higher order could be built following the process described in \cite{Christensen2014}. Such scoring rules would benefit from the interpretability induced by their construction and the ease to be applied to ensemble forecasts. However, they would also inherit the limitation of being only proper.\\

When the probabilistic forecast $F$ has a pdf $f$, scoring rules of a different type can be defined. Let $\calL_\alpha(\bbR)$ denote the class of probability measures on $\bbR$ that are absolutely continuous with respect to $\mu$ (usually taken as the Lebesgue measure) and have $\mu$-density $f$ such that 
\begin{equation*}
    \lVert f\rVert_\alpha = \left( \int_\bbR f(x)^\alpha\mu(\rmd x) \right)^{1/\alpha} <\infty.
\end{equation*}
The most popular scoring rule based on the pdf is the logarithmic score (also known as ignorance score; \citealt{Good1952, Roulston2002}). The logarithmic score is defined as
\begin{equation}\label{eq:log_score}
    \logS(F,y) = -\log(f(y)),
\end{equation}
for $y$ such that $f(y)>0$. In its formulation, the logarithmic score is different from the scoring rules seen previously. \cite{Good1952} proposed the logarithmic score knowing its link with the theory of information: its entropy is the Shannon entropy \citep{Shannon1948} and its expectation is related to the Kullback-Leibler divergence \citep{Kullback1951} (see Appendix~\ref{appendix:expected_scores}). The logarithmic score is strictly proper relative to the class $\calL_1(\bbR)$. Moreover, inference via minimization of the expected logarithmic score is equivalent to maximum likelihood estimation (see, e.g., \citealt{Dawid2015}). The logarithmic score belongs to the family of \textit{local scoring rules}, which are scoring rules only depending on $y$, $f(y)$ and its derivatives up to a finite order. Another local scoring rule is the Hyvärinen score (also known as the gradient scoring rule; \citealt{Hyvaerinen2005}) and it is defined as
\[
    \mathrm{HS}(F,y) = 2\frac{f''(y)}{f(y)} - \frac{f'(y)^2}{f(y)^2},
\]
 for $y$ such that $f(y)>0$. The Hyvärinen score is proper relative to the subclass of $\calP(\bbR)$ such that the density $f$ exists, is twice continuously differentiable and satisfies $f'(x)/f(x)\to 0$ as $|x|\to\infty$. It is worth noticing that the Hyvärinen score can be computed even if $f$ is only known up to a scale factor (e.g., up to a normalizing constant). This property allows circumventing the use of Monte Carlo methods or approximations of the normalizing constant when it is unavailable or hard to compute. This is a property of local proper scoring rules except for the logarithmic score \citep{Parry2012}. Readers eager to learn more about local proper scoring rules may refer to \cite{Parry2012} and \cite{Ehm2012}.

The logarithmic score and the Hyvärinen score do not allow $f$ to be zero. To overcome this limitation, scoring rules expressed in terms of the $L_\alpha$-norm have been proposed. The quadratic score is defined as
\[
    \mathrm{QuadS}(F,y) = \lVert f \rVert^2_2 - 2f(y),
\]
where $\lVert f \rVert^2_2 = \int_\bbR f(y)^2 \rmd y$. The quadratic score is strictly proper relative to the class $\calL_2(\bbR)$.

The pseudospherical score is defined as
\[
    \mathrm{PseudoS}(F,y) = -f(y)^{\alpha-1}/\lVert f\rVert^{\alpha-1}_\alpha,
\]
with $\alpha>1$. For $\alpha=2$, it reduces to the spherical score (see, e.g., \citealt{Jose2007}). The pseudospherical score is strictly proper relative to the class $\calL_\alpha(\bbR)$. The four scoring rules presented above have been criticized as they do not encourage a high probability in the vicinity of the observation $y$ \citep{Gneiting2007}. In particular, as the logarithmic score is more sensitive to outliers, probabilistic forecasts inferred by its minimization may be overdispersive \citep{Gneiting2005}. Moreover, the pdf is not always available, for example in the case of ensemble forecasts.

Readers may refer to the various reviews of scoring rules available (see, e.g., \citealt{Broecker2007, Gneiting2007, Gneiting2014, Thorarinsdottir2018, Alexander2022}). Formulas of the expected scoring rules presented are available in Appendix~\ref{appendix:expected_scores}.\\

Strictly proper scoring rules can be seen as more powerful than proper scoring rules. This is theoretically true when the interest is in identifying the ideal forecast (i.e., the true distribution). Regardless, in practice, scoring rules are also used to rank probabilistic forecasts and with that in mind, a given ranking of forecasts in terms of the expectation of a strictly proper scoring rule (such as the CRPS) is harder to interpret than a ranking in terms of the expectation of a proper but more interpretable scoring rule (such as the SE). The SE is known to discriminate the mean, and thus, a better rank in terms of expected SE implies a better prediction of the mean of the true distribution. Conversely, a better ranking in terms of CRPS implies a better prediction of the whole prediction, but it might not be useful as is, and other verification tools will be needed to know what caused this ranking. When forecasts are not calibrated, there seems to be a trade-off between interpretability and discriminatory power and this becomes more prominent in a multivariate setting. However, simpler interpretable tools and discriminatory-powerful tools can be used complementarily. The framework proposed in Section~\ref{section:framework} aims at helping the construction of interpretable proper scoring rules.\\

\subsection{Multivariate scoring rules}

In a multivariate setting, forecasters cannot solely use univariate scoring rules as they are not able to discriminate forecasts beyond their $1$-dimensional marginals. Univariate scoring rules cannot discriminate the dependence structure between the univariate margins. Multivariate forecasts can be applied in different setups: spatial forecasts, temporal forecasts, multivariable forecasts or any combination of these categories (e.g., spatio-temporal forecasts of multiple variables). Considering weather forecasting, a spatial forecast could aim at predicting temperatures across multiple locations. A temporal forecast could be focused on predicting rainfall at multiple lead times at a given location. A multivariable forecast could predict both eastward and northward components of the wind. In the following, we consider $F\in\calF\subseteq\calP(\bbR^d)$ a multivariate probabilistic forecast and $\bmy\in\bbR^d$ an observation.

Even if there is no natural ordering in the multivariate case, the notions of median and quantile can be adapted using level sets, and then scoring rules using these quantities can be constructed (see, e.g., \citealt{Meng2023}). Nonetheless, as the mean is well-defined, the squared error (SE) can be defined in the multivariate setting :
\begin{equation}\label{eq:ES_multivariate}
    \SE(F,\bmy) = \lVert \bm{\mu}_F-\bmy\rVert^2_2,
\end{equation}
where $\mu_F$ is the mean vector of the distribution $F$. Similar to the univariate case, the SE is proper relative to $\calP_2(\bbR^d)$. Moments are well-defined in the multivariate case allowing the multivariate version of the Dawid-Sebastiani score to be defined. The Dawid-Sebastiani score (DSS) was proposed in \cite{Dawid1999} as
\begin{equation*}
    \DSS(F,\bmy) = \log(\det \Sigma_F) + (\bm{\mu}_F-\bmy)^T\Sigma_F^{-1}(\bm{\mu}_F-\bmy),
\end{equation*}
where $\bm{\mu}_F$ and $\Sigma_F$ are the mean vector and the covariance matrix of the distribution $F$. The DSS is proper relative to $\calP_2(\bbR^d)$ and it becomes strictly proper relative to any convex class of probability measures characterized by their first two moments \citep{Gneiting2007}. The second term in the DSS is the squared Mahalanobis distance between $\bmy$ and $\bm{\mu}_F$.\\

To define a strictly proper scoring rule for multivariate forecast, \cite{Gneiting2007} proposed the energy score (ES) as a generalization of the CRPS to the multivariate case. The ES is defined by
\begin{equation}\label{eq:ES}
    \ES_\alpha(F,\bmy) =\bbE_F\lVert \bmX-\bmy\lVert^\alpha_2-\frac{1}{2}\bbE_F\lVert \bmX-\bmX'\lVert^\alpha_2,
\end{equation}
where $\alpha\in(0,2)$ and $F\in\calP_\alpha(\bbR^d)$, the class of Borel probability measures on $\bbR^d$ such that the moment of order $\alpha$ is finite. The definition of the ES is related to the kernel form of the CRPS \eqref{eq:CRPS}, to which the ES reduces for $d=1$ and $\alpha=1$. As pointed out in \cite{Gneiting2007}, in the limiting case $\alpha=2$, the ES becomes the SE \eqref{eq:ES_multivariate}. The ES is strictly proper relative to $\calP_\alpha(\bbR^d)$ \citep{Szekely2003,Gneiting2007} and is suited for ensemble forecasts \citep{Gneiting2008}. Moreover, the parameter $\alpha$ gives some flexibility: a small value of $\alpha$ can be chosen and still lead to a strictly proper scoring rule, for example, when higher-order moments are ill-defined. The discrimination ability of the ES has been studied in numerous studies (see, e.g., \citealt{Pinson2012, Pinson2013, Scheuerer2015}). \cite{Pinson2012} studied the ability of the ES to discriminate among rival sets of scenarios (i.e., forecasts) of wind power generation. In the case of bivariate Gaussian processes, \cite{Pinson2013} illustrated that the ES appears to be more sensitive to misspecifications of the mean rather than misspecifications of the variance or dependence structure. The lack of sensitivity to misspecifications of the dependence structure has been confirmed in \cite{Scheuerer2015} using multivariate Gaussian random vectors of higher dimension. Moreover, the discriminatory power of the ES deteriorates in higher dimensions \citep{Pinson2013}.

To overcome the discriminatory limitation of the ES, \cite{Scheuerer2015} proposed the variogram score ($\VS$), a score targeting the verification of the dependence structure. The VS of order $p$ is defined as
\begin{equation}\label{eq:VS}
    \VS_p(F,\bmy)=\sum_{i,j=1}^{d} w_{ij}\left(\bbE_F\left[|X_i-X_j|^p\right]-|y_i-y_j|^p\right)^2
\end{equation}
where $X_i$ is the $i$-th component of the random vector $X$ following $F$, $w_{ij}$ are nonnegative weights and $p>0$. The variogram score capitalizes on the variogram, used in spatial statistics to access the dependence structure. The VS cannot detect an equal bias across all components. The VS of order $p$ is proper relative to the class of Borel probability measures on $\bbR^d$ such that the $2p$-th moments of all univariate margins are finite. The weights $w_{ij}$ can be selected to emphasize or depreciate certain pair interactions. For example, in a spatial context, it can be expected the dependence between pairs decays with the distance: choosing the weights proportional to the inverse of the distance between locations can increase the signal-to-noise ratio and improve the discriminatory power of the VS \citep{Scheuerer2015}.\\

When the pdf $f$ of the probabilistic forecast $F$ is available, multivariate versions of the univariate scoring rules based on the pdf are available. The multivariate versions of the scoring rules have the same properties and limitations as their univariate counterpart. The logarithmic score \eqref{eq:log_score} has a natural multivariate version :
\begin{equation*}
    \logS(F,\bmy) = -\log(f(\bmy)),
\end{equation*}
for $\bmy$ such that $f(\bmy)>0$. The logarithmic score is strictly proper relative to the class $\calL_1(\bbR^d)$.

The Hyvärinen score (HS; \citealt{Hyvaerinen2005}) was initially proposed in its multivariate form
\[
    \mathrm{HS}(F,\bmy) = 2 \Delta \log(f(\bmy)) + |\nabla \log(f(\bmy))|^2,
\]
for $\bmy$ such that $f(\bmy)>0$, where $\Delta$ is the Laplace operator (i.e., the sum of the second-order partial derivatives) and $\nabla$ is the gradient operator (i.e., vector of the first-order partial derivatives). In the multivariate setting, the HS can also be computed if the predicted pdf is known up to a normalizing constant. The HS is proper relative to the subclass of $\calP(\bbR^d)$ such that the density $f$ exists, is twice continuously differentiable and satisfies $\lVert\nabla\log(f(x))\rVert\to0$ as $\lVert x\rVert\to \infty$.

The quadratic score and pseudospherical score are directly suited to the multivariate setting :
\begin{align*}
    \mathrm{QuadS}(F,\bmy) &= \lVert f\rVert^2_2 -2f(\bmy);\\
    \mathrm{PseudoS}(F,\bmy) &= -f(\bmy)^{\alpha-1}/\lVert f\rVert^{\alpha-1}_{\alpha},
\end{align*}
where $\lVert f\rVert_\alpha=(\int_{\bbR^d}f(\bmy)^\alpha\rmd \bmy)^{1/\alpha}$. The quadratic score is strictly proper relative to the class $\calL_2(\bbR^d)$. The pseudospherical score is strictly proper relative to the class $\calL_\alpha(\bbR^d)$.

Additionally, other multivariate scoring rules have been proposed among which the marginal-copula score \citep{Ziel2019} or wavelet-based scoring rules (see, e.g., \citealt{Buschow2019}). These scoring rules will be briefly mentioned in Section~\ref{section:example} in light of the proper scoring rule construction framework proposed in this article. Appendix~\ref{appendix:multi_expected_scores} provides formulas for the expected multivariate scoring rules presented above.

\subsection{Spatial verification tools}\label{subsection:spatial_verif}

Spatial forecasts are a very important group of multivariate forecasts as they are involved in various applications (e.g., weather or renewable energy forecasting). Spatial fields are often characterized by high dimensionality and potentially strong correlations between neighboring locations. These characteristics make the verification of spatial forecasts very demanding in terms of discriminating misspecified dependence structures, for example. In the case of spatial forecasts, it is known that traditional verification methods (e.g., gridpoint-by-gridpoint verification) may result in a double penalty. The \textit{double-penalty effect} was pinned in \cite{Ebert2008} and refers to the fact that if a forecast presents a spatial (or temporal) shift with respect to observations, the error made would be penalized twice: once where the event was observed and again where the forecast predicted it.  In particular, high-resolution forecasts are more penalized than less realistic blurry forecasts. The double-penalty effect may also affect spatio-temporal forecasts in general. \\

In parallel with the development of scoring rules, various application-focused spatial verification methods have been developed to evaluate weather forecasts. The efforts toward improving spatial verification methods have been guided by two projects: the intercomparison project (ICP; \citealt{Gilleland2009}) and its second phase, called Mesoscale Verification Intercomparison over Complex Terrain (MesoVICT; \citealt{Dorninger2018}). These projects resulted in the comparison of spatial verification methods with a particular focus on understanding their limitations and clarifying their interpretability. Only a few links exist between the approaches studied in these projects (and the work they induced) and the proper scoring rules framework. In particular, \cite{Casati2022} noted "a lack of representation of novel spatial verification methods for ensemble prediction systems". In general, there is a clear lack of methods focusing on the spatial verification of probabilistic forecasts. Moreover, to help bridging the gap between the two communities, we would like to recall the approach of spatial verification tools in the light of the scoring rule framework introduced above.\\

One of the goals of the ICP was to provide insights on how to develop methods robust to the double-penalty effect. In particular, \cite{Gilleland2009} proposed a classification of spatial verification tools updated later in \cite{Dorninger2018} resulting in a five-category classification. The classes differ in the computing principle they rely on. Not all spatial verification tools mentioned in these studies can be applied to probabilistic forecasts, some of them can solely be applied to deterministic forecasts. In the following description of the classes, we try to focus on methods suited to probabilistic forecasts or at least the special case of ensemble forecasts. 

\textit{Neighborhood}-based methods consist of applying a smoothing filter to the forecast and observation fields to prevent the double-penalty effect. The smoothing filter can take various forms (e.g., a minimum, a maximum, a mean, or a Gaussian filter) and be applied over a given neighborhood. For example, \cite{Stein2022} proposed a neighborhood-based CRPS for ensemble forecasts gathering forecasts and observations made within the neighborhood of the location considered. The use of a neighborhood prevents the double-penalty effect from taking place at scales smaller than that of the neighborhood. In this general definition, neighborhood-based methods can lead to proper scoring rules, in particular, see the notion of patches in Section~\ref{section:example}.

\textit{Scale-separation} techniques denote methods for which the verification is obtained after comparing forecast and observation fields across different scales. The scale-separation process can be seen as several single-bandpass spatial filters (e.g., projection onto a base of wavelets as wavelet-based scoring rules; \citealt{Buschow2019}). However, in order to obtain proper scoring rules, the comparison of the scale-specific characteristics needs to be performed using a proper scoring rule. Section~\ref{section:example} provides a discussion on wavelet-based scoring rules and their propriety. 

\textit{Object-based} methods rely on the identification of objects of interest and the comparison of the objects obtained in the forecast and observation fields. Object identification is application-dependent and can take the form of objects that forecasters are familiar with (e.g., storm cells for precipitation forecasts). A well-known verification tool within this class is the structure-amplitude-location  (SAL; \citealt{Wernli2008}) method which has been generalized to ensemble forecasts in \cite{Radanovics2018}. The three components of the ensemble SAL do not lead to proper scoring rules. They rely on the mean of the forecast within scoring functions inconsistent with the mean. Thus, the ideal forecast does not minimize the expected value. Nonetheless, the three components of the SAL method could be adapted to use proper scoring rules sensitive to the misspecification of the same features.

\textit{Field-deformation} techniques consist of deforming the forecasts field into the observation field (the similarity between the fields can be ensured by a metric of interest). The field of distortion associated with the morphing of the forecast field into the observation field becomes a measure of the predictive performance of the forecast (see, e.g., \citealt{Han2018}).

\textit{Distance measures} between binary images, such as exceedance of a threshold of interest, of the forecast and observation fields. These methods are inspired by development in image processing (e.g., Baddeley's delta measure \citealt{Gilleland2011}).

These five categories are partially overlapping as it can be argued that some methods belong to multiple categories (e.g., some distance measures techniques can be seen as a mix of field-deformation and object-based). They define different principles that can be used to build verification tools that are not subject to the double-penalty effect. The reader may refer to \cite{Dorninger2018} and references therein for details on the classification and the spatial verification methods not used thereafter. The frontier between the aforementioned spatial verification methods and the proper scoring rules framework is porous with, for example, wavelet-based scoring rules belonging to both. It appears that numerous spatial verification methods seek interpretability and we believe that this is not incompatible with the use of proper scoring rules. We propose the following framework to facilitate the construction of interpretable proper scoring rules.\\

\section{A framework for interpretable proper scoring rules}\label{section:framework}

We define a framework to design proper scoring rules for multivariate forecasts. Its definition is motivated by remarks on the multivariate forecasts literature and operational use. There seems to be a growing consensus around the fact that no single verification method has it all (see, e.g., \citealt{Bjerregaard2021}).
Most of the studies comparing forecast verification methods highlight that verification procedures should not be reduced to the use of a single method and that each procedure needs to be well suited to the context (see, e.g., \citealt{Scheuerer2015, Thorarinsdottir2018}). Moreover, from a more theoretical point of view, (strict) propriety does not ensure discrimination ability and different (strictly) proper scoring rules can lead to different rankings of sub-efficient forecasts.

Standard verification procedures gradually increase the complexity of the quantities verified. Procedures often start by verifying simple quantities such as quantiles, mean, or binary events (e.g., prediction of dry/wet events for precipitation). If multiple forecasts have a satisfying performance for these quantities, marginal distributions of the multivariate forecast can be verified using univariate scoring rules. Finally, multivariate-related quantities, such as the dependence structure, can be verified through multivariate scoring rules. Forecasters rely on multiple verification methods to evaluate a forecast and ideally, the verification method should be interpretable by targeting specific aspects of the distribution or thanks to the forecaster's experience. This type of verification procedure allows the forecaster to understand what characterizes the predictive performance of a forecast instead of directly looking at a strictly proper scoring rule giving an encapsulated summary of the predictive performance.

Various multivariate forecast calibration methods rely on the calibration of univariate quantities obtained by dimension reduction techniques. As the general principle of multivariate calibration leans on studying the calibration of quantities obtained by pre-rank functions, \cite{Allen2024Assessing} argue that calibration procedures should not rely on a single pre-rank function and should instead use multiple simple pre-rank functions and leverage the interpretability of the PIT/rank histograms associated. A similar principle can be applied to increase the interpretability of verification methods based on scoring rules.

As general multivariate strictly proper scoring rules fail to discriminate forecasts with respect to arbitrary misspecifications and they may lead to different ranking of sub-efficient forecasts, multivariate verification could benefit from using multiple proper scoring rules targeting specific aspects of the forecasts. Thereby, forecasters know which aspect of the observations are well-predicted by the forecast and can update their forecast or select the best forecast among others in the light of this better understanding of the forecast. To facilitate the construction of interpretable proper scoring rules, we define a framework based on two principles: transformation and aggregation.\\

The transformation principle consists of transforming both forecast and observation before applying a scoring rule. \cite{HeinrichMertsching2021} introduced this general principle in the context of point processes. In particular, they present scoring rules based on summary statistics targeting the clustering behavior or the intensity of the processes. In a more general context, the use of transformations was disseminated in the literature for several years (see Section~\ref{section:example}). Proposition~\ref{prop:transformation_SR} shows how transformations can be used to construct proper scoring rules.

\begin{proposition}\label{prop:transformation_SR}
   Let $\calF\subset\calP(\bbR^d)$ be a class of Borel probability measure on $\bbR^d$ and let $F\in\calF$ be a forecast and $\bmy\in\bbR^d$ an observation. Let $T:\bbR^d\to\bbR^k$ be a transformation and let $\rmS$ be a scoring rule on $\bbR^k$ that is proper relative to $T(\calF)=\{\calL(T(\bmX)), \bmX\sim F\in\calF\}$. Then, the scoring rule
    \[
        \rmS_{T}(F,\bmy) = \rmS(T(F),T(\bmy))
    \]
    is proper relative to $\calF$. If $\rmS$ is strictly proper relative to $T(\calF)$ and $T$ is injective, then the resulting scoring rule $\rmS_T$ is strictly proper relative to $\calF$.
\end{proposition}

To gain interpretability, it is natural to have dimension-reducing transformations (i.e., $k<d$), which generally leads to $T$ not being injective and $\rmS_T$ not being strictly proper. Nonetheless, as expressed previously, interpretability is important and it can mostly be leveraged if the transformation simplifies the multivariate quantities. Particularly, it is generally preferred to choose $k=1$ to make the quantity easier to interpret and focus on specific information contained in the forecast or the observation. Straightforward transformations can be projections on a $k$-dimensional margin or a summary statistic relevant to the forecast type such as the total over a domain in the case of precipitations. Simple transformations may be preferred for their interpretability and their potential lack of discriminatory power can be made up for via the use of multiple simpler transformations. Numerous examples of transformations are presented, discussed, and linked to the literature in Section~\ref{section:example}. The proof of Proposition~\ref{prop:transformation_SR} is provided in Appendix~\ref{appendix:aggregation-transformation}.\\

The second principle is the aggregation of scoring rules. Aggregation can be used on scoring rules in order to combine them and obtain a single scoring rule summarizing the evaluation. It can be used to operate on scoring rules acting on different spaces, times or locations. Note that \cite{Dawid2014} introduced the notion of \textit{composite score} which is related to the aggregation principle but is closer to the combined application of both principles. Proposition~\ref{prop:aggregation_SR} presents a general aggregation principle to build proper scoring rules. This principle has been known since proper scoring rules have been introduced.

\begin{proposition}\label{prop:aggregation_SR}
    Let $\calS=\{\rmS_i\}_{1\leq i\leq m}$ be a set of proper scoring rules relative to $\calF\subset\mathcal{P}(\bbR^d)$. Let $\bm{w}=\{w_i\}_{{1\leq i\leq m}}$ be nonnegative weights. Then, the scoring rule
    \begin{equation*}
        \rmS_{\calS,\bm{w}}(F,\bmy) = \sum_{i=1}^m w_i \rmS_i(F,\bmy)
    \end{equation*}
    is proper relative to $\calF$. If at least one scoring rule $\rmS_i$ is strictly proper relative to $\calF$ and $w_i>0$, the aggregated scoring rule $\rmS_{\calS,\bm{w}}$ is strictly proper relative to $\calF$.
\end{proposition}

 It is worth noting that Proposition~\ref{prop:aggregation_SR} does not specify any strict condition for the scoring rules used. For example, the scoring rules aggregated do not need to be the same or do not need to be expressed in the same units. Aggregated scoring rules can be used to summarize the evaluation of univariate probabilistic forecasts (e.g., aggregation of CRPS at different locations) or to summarize complementary scoring rules (e.g., aggregation of Brier score and a threshold-weighted CRPS). Unless stated otherwise, for simplicity, we will restrict ourselves to cases where the aggregated scoring rules are of the same type. \cite{Bolin2023} showed that the aggregation of scoring rules can lead to unintuitive behaviors. For the aggregation of univariate scoring rules, they showed that scoring rules do not necessarily have the same dependence on the scale of the forecasted phenomenon: this leads to scoring rules putting more (or less) emphasis on the forecasts with larger scales. They define and propose local scale-invariant scoring rules to make scale-agnostic scoring rules. When performing aggregation, it is important to be aware of potential preferences or biases of the scoring rules. 

We only consider aggregation of proper scoring rules through a weighted sum. To conserve (strict) propriety of scoring rules, aggregations can take, more generally, the form of (strictly) isotonic transformations, such as a multiplicative structure when positive scoring rules are considered \citep{Ziel2019}. \\

The two principles of Proposition~\ref{prop:transformation_SR} and Proposition~\ref{prop:aggregation_SR} can be used simultaneously to create proper scoring rules based on both transformations and aggregation as presented in Corollary~\ref{cor:agg+transform_SR}.

\begin{corollary}\label{cor:agg+transform_SR}
    Let $\calT=\{T_i\}_{{1\leq i\leq m}}$ be a set of transformations from $\bbR^d$ to $\bbR^k$. Let $\calS_\calT=\{{\rmS_{T_i}}\}_{{1\leq i\leq m}}$ be a set of proper scoring rules where $\rmS$ is proper relative to $T_i(\calF)$, for all $1\leq i\leq m$. Let $\bm{w}=\{w_i\}_{{1\leq i\leq m}}$ be nonnegative weights. Then, the scoring rule
    \begin{equation*}
        \rmS_{\calS_\calT,\bm{w}}(F,\bmy) = \sum_{i=1}^m w_i \rmS_{T_i}(F,\bmy)
    \end{equation*}
    is proper relative to $\calF$.
\end{corollary}

Strict propriety relative to $\calF$ of the resulting scoring rule is obtained as soon as there exists $1\leq i\leq m$ such that $\rmS$ is strictly proper relative to $T_i(\calF)$, $T_i$ is injective and $w_i>0$. The result of Corollary~\ref{cor:agg+transform_SR} can be extended to transformations with images in different dimensions and paired with different scoring rules (see Appendix~\ref{appendix:general_cor}).\\

As we will see in the examples developed in the following section, numerous scoring rules used in the literature are based on these two principles of aggregation and transformation.

\bigskip\noindent
\paragraph{Decomposition of kernel scoring rules.}
We briefly discuss the link between the transformation and aggregation principles for scoring rules and the specific class of kernel scoring rules. A kernel on $\bbR^d$ is a measurable function $\rho:\bbR^d\times\bbR^d\to\bbR$ satisfying the following two properties:
\begin{itemize}
    \item[$i)$] (symmetry) $\rho(\bmx_1,\bmx_2)=\rho(\bmx_2,\bmx_1)$ for all $\bmx_1,\bmx_2\in\bbR^d$;
    \item[$ii)$] (non-negativity)  $\sum_{1\leq i\leq j\leq n} a_i a_j \rho(\bmx_i,\bmx_j)\geq 0$ for all $\bmx_1,\ldots,\bmx_n\in\bbR^d$ and $a_1,\ldots,a_n\in\bbR$, for all $n\in\bbN$.
\end{itemize}
The kernel scoring rule $\rmS_\rho$ associated with the kernel $\rho$ is defined on the space of predictive distributions 
\[
    \mathcal{P}_\rho=\left\{F\in \mathcal{P}(\bbR^d)\colon \int \sqrt{\rho(x,x)}F(\rmd x)<+\infty\right\}
\]
by
\begin{equation}\label{eq:def-S_k}
    \rmS_\rho(F,\bmy)=\bbE_{F}[\rho(\bmX,\bmy)]-\frac{1}{2}\bbE_{F}[\rho(\bmX,\bmX')]-\frac{1}{2}\rho(\bmy,\bmy),
\end{equation}
where $\bmy\in\bbR^d$ and $\bmX, \bmX'$ are independent random variables following $F$. Importantly,  $\rmS_\rho$ is  proper on $\mathcal{P}_\rho$ and, for an ensemble forecast $F=\frac{1}{M}\sum_{m=1}^M \delta_{\bmx_m}$ with $M$ members $\bmx_1,\ldots,\bmx_M$, it takes the simple form
\begin{equation}\label{eq:def-S_k-ensemble}
    \rmS_\rho(F,\bmy)= \frac{1}{M}\sum_{m=1}^M\rho(\bmx_m,\bmy)-\frac{1}{2M^2}\sum_{1\leq m_1,m_2\leq M}\rho(\bmx_{m_1},\bmx_{m_2})-\frac{1}{2}\rho(\bmy,\bmy), 
\end{equation}
making scoring rules particularly useful for ensemble forecasts. 

The CRPS is surely the most widely used kernel scoring rule. Equation~\eqref{eq:CRPS} shows that it is a associated with the kernel $\rho(x_1,x_2)=|x_1|+|x_2|-|x_1-x_2|$ (the function $|x_1-x_2|$ is conditionally semi-definite negative so that $\rho$ is non-negative). For more details on kernel scoring rules, the reader should refer to \cite{Gneiting2005} or \cite{Steinwart2021}.\\

The following proposition reveals that a kernel scoring rule can always be expressed as an aggregation of squared errors (SEs) between transformations of the forecast-observation pair.
\begin{proposition}\label{prop:series-representation}
    Let $\rmS_\rho$ be the kernel scoring rule associated with the kernel $\rho$. Then there exists a sequence of transformations $T_l:\bbR^d\to\bbR$, $l\geq 1$, such that 
    \[
    \rmS_\rho(F,\bmy)=\frac{1}{2}\sum_{l\geq 1} \SE(T_l(F),T_l(\bmy)),
    \]
    for all predictive distribution $F\in\mathcal{P}_\rho$ and observation $\bmy\in\bbR^d$.
\end{proposition}
In particular, the series on the right-hand side is always finite. The proof is provided in Appendix~\ref{appendix:proof-kernel} and relies on the reproducing kernel Hilbert space (RKHS) representation of kernel scoring rules. In particular, we will see that the sequence $(T_l)_{l\geq 1}$ can be chosen as an orthonormal basis of the RKHS associated with the kernel $\rho$.

This representation of kernel scoring rules can be useful to understand more deeply the comparison of the predictive forecast $F$ and observation $\bmy$. While the definition \eqref{eq:def-S_k} is quite abstract, the series representation can be rewritten
\[
    \rmS_\rho(F,\bmy)=\sum_{l\geq 1} \big(\bbE_F[T_l(\bmX)]-T_l(\bmy)\big)^2
\]
with $X$ a random variable following $F$. In other words, for $l\geq 1$, the observed value $T_l(\bmy)$ is compared to the predicted value $T_l(\bmX)$ under the predictive distribution $F$ using the SE; then all these contributions are aggregated in a series forming the kernel scoring rule. \\

To give more intuition, we study two important cases in dimension $d=1$. The details of the computations are provided in Appendix~\ref{appendix:proof-kernel-examples}. For the Gaussian kernel scoring rule associated with the kernel 
\[
    \rho(x_1,x_2)=\exp(-(x_1-x_2)^2/2),
\]
some computations yield the series representation
\[
    \rmS_\rho(F,y)=\frac{1}{2}\sum_{l\geq 0}\frac{1}{l!}\Big(\bbE_F[X^le^{-X^2/2}]-y^le^{-y^2/2}\Big)^2 
\]
so that this score compares the probabilistic forecast $F$ and the observation $y$ through the transforms
\[
    T_l(x)=\frac{1}{\sqrt{l!}}x^le^{-x^2/2},\quad l\geq 0.
\]

For the CRPS, a possible series representation is obtained thanks to the following wavelet basis of functions: let $T^0(x)=x\mathds{1}_{[0,1)}(x)+\mathds{1}_{[1,+\infty)}(x)$  (plateau function) and $T^1(x)=\big(1/2-|x-1/2|\big)\mathds{1}_{[0,1]}(x)$ (triangle function) and consider the collection of functions
\[
    T^0_l(x)=T^0(x-l),\quad T^1_{l,m}(x)=2^{-m/2}T^1(2^{m}x-l),\quad l\in\bbZ,m\geq 0,
\]
where $l\in\bbZ$ is a position parameter and $m\geq 0$ a scale parameter. Then, the CRPS can be written as
\begin{align*}
    \CRPS(F,y) &=\sum_{l\in\bbZ}\SE(T^0_{l}(F),T^0_{l}(y))+\sum_{l\in\bbZ}\sum_{m\geq 0} \SE(T^1_{l,m}(F),T^1_{l,m}(y))\\
    &=\sum_{l\in\bbZ} \Big(\bbE_F[T^0(X-l)]-T^0(y-l)\Big)^2+\sum_{l\in\bbZ}\sum_{m\geq 0} 2^{-m}\Big(\bbE_F[T^1(2^{m}X-l)]-T(2^{m}y-l)\Big)^2.
\end{align*}
We can see that the CRPS compares forecast and observation through the SE after applying the plateau and triangle transformations for multiple positions and scales and then aggregates all the contributions.

\section{Applications of the transformation and aggregation principles}\label{section:example}

\subsection{Projections}

Certainly, the most direct type of transformation is projections of forecasts and observations on their $k$-dimensional marginals. We denote $T_i$ the projection on the $i$-th component such that $T_i(\bmX)=X_i$, for all $\bmX\in\bbR^d$. This allows the forecaster to assess the predictive performance of a forecast for a specific univariate marginal independently of the other variables.  If $\rmS$ is an univariate scoring rule proper relative to $\calP(\bbR)$, then Proposition~\ref{prop:transformation_SR} leads to $\rmS_{T_i}$ being proper relative to $\calP(\bbR^d)$. This "new" scoring rule can be useful if a given marginal is of particular interest (e.g., location of high interest in a spatial forecast). However, it can be more interesting to aggregate such scoring rules across all $1$-dimensional marginals. This leads to the following scoring rule
\[
    \rmS_{\calS_\calT,\bm{w}}(F,\bmy) = \sum_{i=1}^d w_i \rmS_{T_i}(F,\bmy),
\]
where $\calS_\calT$ is $\{\rmS_{T_i}\}_{{1\leq i\leq d}}$. This setting is popular for assessing the performance of multivariate forecasts and we briefly present examples from the literature falling under this setting. Aggregation of CRPS \eqref{eq:CRPS} across locations and/or lead times is common practice for plots or comparison tables with uniform weights \citep{Gneiting2005, Taillardat2016, Rasp2018, Schulz2022, Lerch2022, Hu2023} or with more complex schemes such as weights proportional to the cosine of the latitude \citep{BenBouallegue2024PoET}. The SE \eqref{eq:SE} and AE \eqref{eq:AE} can be aggregated to obtain RMSE and MAE, respectively \citep{DelleMonache2013, Gneiting2005, Lerch2022, Pathak2022}. \cite{Bremnes2019} aggregated QSs \eqref{eq:QS} across stations and different quantile levels of interest with uniform weights. Note that the multivariate SE \eqref{eq:ES_multivariate} can be rewritten as the sum of univariate SE across $1$-marginals: $\SE(F,\bmy) = \lVert \bm{\mu_F}-\bmy\rVert^2_2 = \sum_{i=1}^d \SE_{T_i}(F,\bmy)$.\\

The second simplest choice is the $2$-dimensional case, allowing to focus on pair dependency. We denote $T_{(i,j)}$ the projection on the $i$-th and $j$-th components (i.e., the $(i,j)$ pair of components) such that $T_{(i,j)}(\bmX)=X_{i,j}=(X_i,X_j)$. In this setting, $\rmS$ has to be a bivariate proper scoring rule to construct a proper scoring rule $\rmS_{T_{(i,j)}}$. 
The aggregation of such scoring rules becomes
\[
    \rmS_{\calS_\calT,\bm{w}}(F,\bmy) = \sum_{\substack{i,j=1\\ i\neq j}}^d w_{i,j} \rmS_{T_{(i,j)}}(F,\bmy).
\]
As suggested in \cite{Scheuerer2015} for the VS \eqref{eq:VS}, the weights $w_{i,j}$ can be chosen appropriately to optimize the signal-to-noise ratio. For example, in a spatial setting where the dependence between locations is believed to decrease with the distance separating them, the weights $w_{i,j}$ can be chosen to be proportional to the inverse of the distance. This bivariate setting is less used in the literature, we present two articles using or mentioning scoring rules within this scope. In a general multivariate setting, \cite{Ziel2019} suggests the use of a marginal-copula scoring rule where the copula score is the bivariate copula energy score (i.e., the aggregation of the energy scores across all the regularized pairs). To focus on the verification of the temporal dependence of spatio-temporal forecasts, \cite{BenBouallegue2024PoET} uses the bivariate energy score over consecutive lead times.\\

In a more general setup, we consider projection on $k$-dimensional marginals. In order to reduce the number of transformation-based scores to aggregate, it is standard to focus on localized marginals (e.g., belonging to patches of a given spatial size). Denote $\calP=\{P_i\}_{1\leq i\leq m}$ a set of valid patches (for some criterion or of a given size) and $\calS_\calP$ the set of transformation-based scores associated with the projections on the patches $\calP$. Given a multivariate scoring rule $\rmS$ proper relative to $\calP(\bbR^k)$, we can construct the following aggregated score :
\[
    \rmS_{\calS_\calP,\bm{w}}(F,\bmy)=\sum_{P\in\calP} w_{P} \rmS_{P}(F,\bmy).
\]
This construction can be used to create a scoring rule only considering the dependence of localized components, given that the patches are defined in that sense. The use of patches has similar benefits as the weighting of pairs given a belief on their correlations: obtain a better signal-to-noise ratio and improve the discrimination of the resulting scoring rule. For example, \cite{Pacchiardi2024} introduced patched energy scores as scoring rules to minimize in order to train a generative neural network. The patched energy scores are defined for $\rmS=\ES$ and square patches spaced by a given stride. Even though spatial patches may be more intuitive, it is possible to use temporal or spatio-temporal patches. Patch-based scoring rules appear as a natural member of the neighborhood-based methods of the spatial verification classification mentioned in Section~\ref{subsection:spatial_verif}. Given that the patches are correctly chosen (e.g., of a size appropriate to the problem at hand), patch-based scoring rules are not subject to the double-penalty effect.\\

As noticeable by the low number of examples available in the literature, aggregation (and plain use) of scoring rules based on projection in dimension $k\geq 2$ is not standard practice, probably because such projections may lack interpretability. Instead, to assess the multivariate aspects of a forecast, scoring rules relying on summary statistics are often favored.

\subsection{Summary statistics}

Summary statistics are a central tool of statisticians' toolboxes as they provide interpretable and understandable quantities that can be linked to the behavior of the phenomenon studied. Moreover, their interpretability can be enhanced by the forecaster's experience and this can be leveraged when constructing scoring rules based on them. Summary statistics are commonly present during the verification procedure and this can be extended by the use of new scoring rules derived from any summary statistic of interest. For example, numerous summary statistics can come in handy when studying precipitations over a region covered by gridded observation and forecasts. Firstly, it is common practice to focus on binary events such as the exceedance of a threshold (e.g., the presence or absence of precipitation). This can be studied by using the BS \eqref{eq:BS} on all $1$-dimensional marginals as mentioned in the previous subsection but also in a multivariate manner through the fraction of threshold exceedances (FTE) over patches as presented further. Regarding precipitations, it is standard to be interested in the prediction of total precipitation over a region or a time period. This transformation of the field can be leveraged to construct a scoring rule. Finally, it is important to verify that the spatial structure of the forecast matches the spatial structure of observations. The spatial structure can be (partially) summarized by the variogram or by wavelet transformations. The predictive performance for the spatial structure can be assessed by their associated scoring rules: the VS of order $p$ \eqref{eq:VS} and the wavelet-based score \citep{Buschow2019}. Other summary statistics can be of interest to the phenomenon studied, \cite{HeinrichMertsching2021} present summary statistics specific to point processes focusing on clustering and intensity.\\

The most well-known summary statistic is certainly the mean. In spatial statistics, it can be used to avoid double penalization when we are less interested in the exact location of the forecast but rather in a regional prediction. The transformation associated with the mean is
\begin{equation}\label{eq:mean_P}
    \mathrm{mean}_P(\bmX) = \frac{1}{|P|}\sum_{i\in P} X_i,
\end{equation}
where $P$ denotes a patch and $|P|$ its dimension. Proposition~\ref{prop:transformation_SR} ensures that this transformation can be used to construct proper scoring rules. The scoring rule involved in the construction has to be univariate, however, the choice depends on the general properties preferred. For example, the SE would focus on the mean of the transformed quantity, whereas the AE would target its median. It is worth noting that the total can be derived by the mean transformation by removing the prefactor
\[
    \mathrm{total}_P(\bmX) = \sum_{i\in P} X_i.
\]
In the case of precipitation, the total is more used than the mean since the total precipitation over a river basin can be decisive in evaluating flood risk. For example, one could construct an adapted version of the amplitude component of the SAL method \citep{Wernli2008, Radanovics2018} using the SE if the mean total precipitation is of interest. \cite{Gneiting2011} presents other links between the quantity of interest and the scoring rule associated. Similarly, the transformations associated with the minimum and the maximum over a patch $P$ can be obtained :
\begin{align*}
    \mathrm{min}_P(\bmX) &= \min_{i\in P}(X_i);\\
    \mathrm{max}_P(\bmX) &= \max_{i\in P}(X_i).
\end{align*}
The maximum or minimum can be useful when considering extreme events. It can help understand if the severity of an event is well-captured. For example, as minimum and maximum temperatures affect crop yields (see, e.g., \citealt{Agnolucci2020}), it can be of particular interest that a weather forecast within an agricultural model correctly predicts the minimum and maximum temperatures. After studying the mean, it is natural to think of the moments of higher order. We can define the transformation associated with the variance over a patch $P$ as
\[
    \mathrm{Var}_P(\bmX) = \frac{1}{|P|}\sum_{i\in P} (X_i-\mathrm{mean}_P(\bmX))^2.
\]
The variance transformation can provide information on the fluctuations over a patch and be used to assess the quality of the local variability of the forecast. In a more general setup, it can be of interest to use a transformation related to the moment of order $n$ and the transformation associated follows naturally
\[
    \mathrm{M}_{n,P}(\bmX) = \frac{1}{|P|}\sum_{i\in P} X_i^n.
\]
More application-oriented transformations are the
central or standardized moments (e.g., skewness or kurtosis). Their transformations can be obtained directly from estimators. As underlined in \cite{HeinrichMertsching2021}, since Proposition~\ref{prop:transformation_SR} applies to any transformation, there is no condition on having an unbiased estimator to obtain proper scoring rules.\\

Threshold exceedance plays an important role in decision making such as weather alerts. For example, MeteoSwiss' heat warning levels are based on the exceedance of daily mean temperature over three consecutive days \citep{Allen2023Weighted}. They can be defined by the simultaneous exceedance of a certain threshold and the fraction of threshold exceedance (FTE) is the summary statistic associated. 
\begin{equation}\label{eq:fte}
    \mathrm{FTE}_{P,t}(\bmX) = \frac{1}{|P|}\sum_{i\in P} \mathds{1}_{\{X_i\geq t\}}.
\end{equation}
FTEs can be used as an extension of univariate threshold exceedances and it prevents the double-penalty effect. FTEs may be used to target compound events (e.g., the simultaneous exceedances of a threshold at multiple locations of interest). \cite{Roberts2008} used an FTE-based SE over different sizes of neighborhoods (patches) to verify at which scale forecasts become skillful. To assess extreme precipitation forecasts, \cite{Rivoire2023} introduces scores for extremes with temporal and spatial aggregation separately. Extreme events are defined as values higher than the seasonal $95\%$ quantile. In the subseasonal-to-seasonal range, the temporal patches are 7-day windows centered on the extreme event and the spatial patches are square boxes of 150~km $\times$ 150~km centered on the extreme event. The final scores are transformed BS \eqref{eq:BS} with a threshold of one event predicted across the patch.\\

Correctly predicting the structure dependence is crucial in multivariate forecasting. Variograms are summary statistics representing the dependence structure. The variogram of order $p$ of the pair $(i,j)$ corresponds to the following transformation :
\[
    \gamma_{ij}^p(\bmX)=|X_i-X_j|^p.
\]
As mentioned in the Introduction, using both the transformation and aggregation principles, we can recover the VS of order $p$ \eqref{eq:VS} introduced in \cite{Scheuerer2015} :
\[
    \VS_p(F,\bmy) = \sum_{i,j=1}^d w_{ij} \SE_{\gamma_{ij}^p}(F,\bmy) = \sum_{i,j=1}^d w_{ij} \left(\bbE_F[|X_i-X_j|^p]-|y_i-y_j|\right)^2.
\]
Along with the well-known VS of order $p$, \cite{Scheuerer2015} introduced alternatives where the scoring rule applied on the transformation is the CRPS \eqref{eq:CRPS} or the AE \eqref{eq:AE} instead of the SE \eqref{eq:SE}. As mentioned previously, under the \textit{intrinsic hypothesis} of \cite{Matheron1963} (i.e., pairwise differences only depend on the distance between locations), the weights can be selected to obtain an optimal signal-to-noise ratio. Moreover, the weights could be selected to investigate a specific scale by giving a non-zero weight to pairs separated by a given distance. 

In the case of spatial forecasts over a grid of size $d\times d$, a spatial version of the variogram transformation is available :
\[
    \gamma_{\bm{i},\bm{j}}(\bmX) = |X_{\bm{i}}-X_{\bm{j}}|^p,
\]
where $\bm{i},\bm{j}\in\calD=\{1,\dots,d\}^2$ are the coordinates of grid points. Under the intrinsic hypothesis of \cite{Matheron1963}, the variogram between grid points separated by the vector $\bm{h}$ can be estimated by :
\[
    \gamma_{\bmX}(\bm{h}) = \frac{1}{2|\calD(\bm{h})|} \sum_{\bm{i}\in\calD(\bm{h})} \gamma_{\bm{i},\bm{i}+\bm{h}}(\bmX),
\]
where $\calD(\bm{h})=\{\bm{i}\in\calD:\bm{i}+\bm{h}\in\calD\}$. This directed variogram can be used to target the verification of the anisotropy of the dependence structure. The isotropy transformation associated to the distance $h$ can be defined by
\begin{equation}\label{eq:T_iso}
    T_{\mathrm{iso},h}(\bmX) = - \cfrac{\big(\gamma_X((h,0))-\gamma_X((0,h))\big)^2}{\cfrac{2\gamma_X((h,0))^2}{|\calD((h,0))|}+\cfrac{2\gamma_X((0,h))^2}{|\calD((0,h))|}}.
\end{equation}
This transformation is the isotropy pre-rank function proposed in \cite{Allen2024Assessing}. The isotropy transformation considers the orthogonal directions formed by the abscissa and ordinate axes and evaluates how the variogram changes between these directions. The transformation leads to negative or zero quantities with values close to zero characterizing isotropy and negative values corresponding to the anisotropy of the variograms in the directions and at the scale involved.

\subsection{Other transformations}

Transformations other than projections or summary statistics can be used to target forecast characteristics. For example, a transformation in the form of a change of coordinates or a change of scale (e.g., a logarithmic scale) can be used to obtain proper scoring rules. We highlight two families of scoring rules that can be seen as transformation-based scoring rules: wavelet-based scoring rules and threshold-weighted scoring rules.\\

Generally speaking, wavelet-based scoring rules are built thanks to a projection of forecast and observation fields onto a wavelet basis. Based on the wavelet coefficients, dimension reduction might be performed to target specific characteristics such as the dependence structure or the location. The resulting coefficients of the forecast fields are compared to the coefficients of the observations fields using scoring rules (e.g., squared error (SE) or energy score (ES)). Wavelet transformations are (complex) transformations, and thus, the scoring rules associated fall within the scope of Proposition~\ref{prop:transformation_SR}. In particular, \cite{Buschow2019} used a dimension reduction procedure resulting in the obtention of a mean and a scale spectra and used scoring rules to compare forecasts and observation spectra. For example, the ES of the mean spectrum is used and shows good discrimination ability when the scale structure is misspecified. 

Note that \cite{Buschow2019} proposed two other wavelet-based scoring rules: one based on the earth mover's distance (EMD) of the scale histograms and one based on the distance in the scale histograms' center of mass. The EMD-based scoring rules are not proper since the EMD is not a proper scoring rule \citep{Thorarinsdottir2013} and the so-called distance between centers of mass is not a distance but rather a difference of position leading to an improper scoring rule. However, the ES-based scoring rules are proper and could be derived from scale histograms. Despite their apparent complexity, wavelet transformations allow to target interpretable characteristics such as the location \citep{Buschow2022}, the scale structure \citep{Buschow2019, Buschow2020} or the anisotropy \citep{Buschow2021}. The transformations proposed for the deterministic forecasts setting in most of these articles could be used as foundations for future work willing to propose wavelet-based proper scoring rules targeting the location, the scale structure or the anisotropy.\\

As showcased in \cite{HeinrichMertsching2021} for a specific example and hinted in \cite{Allen2024Assessing}, transformations can also be used to emphasize certain outputs. Threshold weighting is one of the three main types of weighting conserving the propriety of scoring rules. Its name come from the fact that it corresponds to a weighting over different thresholds in the case of CRPS \eqref{eq:CRPS_bs} \citep{Gneiting2011}. Recall that given a conditionally negative definite kernel $\rho$, the kernel scoring associated $\rmS_\rho$ \eqref{eq:def-S_k} is proper relative to $\calP_\rho$. Many popular scoring rules are kernel scores such as the BS \eqref{eq:BS}, the CRPS \eqref{eq:CRPS}, the ES \eqref{eq:ES} and the VS \eqref{eq:VS}. By definition \citep[Definition 4]{Allen2023Evaluating}, threshold-weighted kernel scores are constructed as
\begin{align*}
    \mathrm{tw}\rmS_\rho(F,\bmy;v) &= \bbE_F[\rho(v(\bmX),v(\bmy))]-\frac{1}{2}\bbE_F[\rho(v(\bmX),v(\bmX'))]-\frac{1}{2}\rho(v(\bmy),v(\bmy));\\
    &= \rmS_\rho(v(F),v(\bmy)),
\end{align*}
where $v$ is the chaining function capturing how the emphasis is put on certain outputs. With this explicit definition, it is obvious that threshold-weighted kernel scores are covered by the framework of Proposition~\ref{prop:transformation_SR}. It can be noted that Proposition~4 in \cite{Allen2023Evaluating} states that strict propriety of the kernel scoring rule is preserved by the chaining function $v$ if and only if $v$ is injective. Weighted scoring rules allow to emphasize particular outcomes: when studying extreme events, it is often of particular interest to focus on values larger than a given threshold $t$ and this can be achieved using the chaining rule $v(x)=\mathds{1}_{x\geq t}$. Threshold-weighted scoring rules have been used in verification procedures in the literature; we illustrate its use through three different studies. \cite{Lerch2013} aggregated across stations twCRPS to compare the upper tail performance of different daily maximum wind speed forecasts. \cite{Chapman2022} aggregated the threshold-weighted CRPS across locations to study the improvement of statistical postprocessing techniques, the importance of predictors and the influence of the size of the training set on the performance. \cite{Allen2023Weighted} used threshold-weighted versions of the CRPS, the ES, and the VS to compare the predictive performance of forecasts regarding heatwave severity; the scoring rules were aggregated across stations. Readers may refer to \cite{Allen2023Weighted} and \cite{Allen2023Evaluating} for insightful reviews of weighted scoring rules in both univariate and multivariate settings.\\

\section{Simulation study}\label{section:sim-study}

This section provides simulated examples to showcase the different uses of the framework introduced in Section~\ref{section:framework} to construct interpretable proper scoring rules for multivariate forecasts. Four examples are developed. Firstly, a setup where the emphasis is put on $1$-marginal verification is proposed. This setup serves as a means of recalling and showing the limitations of strictly proper scoring rules and the benefits of interpretable scoring rules in a concrete setting. Secondly, a standard multivariate setup is studied where popular multivariate scoring rules (i.e., VS and ES) are compared to a multivariate scoring rule aggregated over patches and an aggregation-and-transformation-based scoring rule in their discrimination ability regarding the dependence structure. Thirdly, a setup introducing anisotropy in both observations and forecasts is introduced. The anisotropic score is constructed based on the transformation principle with the goal of discriminating differences of anisotropy in the dependence structure between forecast and observations. Fourthly, we propose a setup to test the sensitivity of scoring rules to the double-penalty effect and we introduce scoring rules that can be built to be resilient to some manifestation of the double-penalty effect. 

In these four numerical experiments, the spatial field is observed and predicted on a regular $20\times20$ grid  $\calD=\{1,\ldots,20\}\times\{1,\ldots,20\}$. Observations are realizations of a Gaussian random field $(G(s))_{s\in\calD}$ with zero mean and power-exponential covariance  defined as
\begin{equation}\label{eq:obs_simu}
    \mathrm{cov}(G(s),G(s')) = {\sigma_0}^2 \exp\left(-\left(\frac{\lVert s-s'\rVert}{\lambda_0}\right)^{\beta_0}\right),\quad s,s'\in\calD.
\end{equation}
The parameters are taken equal to $\sigma_0=1$, $\lambda_0=3$ and $\beta_0=1$.

In each numerical experiment, we compare a few predictive distributions, including the distribution generating observations and other ones deviating from the generative distributions in a specific way. These different predictive distributions are evaluated with different scoring rules and the aim is to illustrate the discriminatory ability of the different scoring rules.

The simulation study uses 500 observations of the random field $(G(s))_{s\in\calD}$. The scoring rules are computed using exact formulas when possible (see Appendix~\ref{appendix:sr_simu}), and, when exact formulas are not available, they are computed based on a sample of size 100 (i.e., ensemble forecasts) of the probabilistic forecast. Estimated expectations over the 500 observations are computed and the experiment is repeated 10 times. The corresponding results are represented by boxplots. The units of the scoring rules are rescaled by the average expected score of the true distribution (i.e., the ideal forecast). The statistical significativity of the ranking between forecasts is tested using a Diebold-Mariano test \citep{Diebold1995}. When deemed necessary, statistical significativity is mentioned for a confidence level of 95\%.

The code used for the different numerical experiments is publicly available\footnote{\url{https://github.com/pic-romain/aggregation-transformation}}.

\subsection{Marginals}

\begin{figure}[ht]
    \begin{subfigure}[c]{.3\textwidth}
        \begin{flushleft}
            \includegraphics[scale=.57]{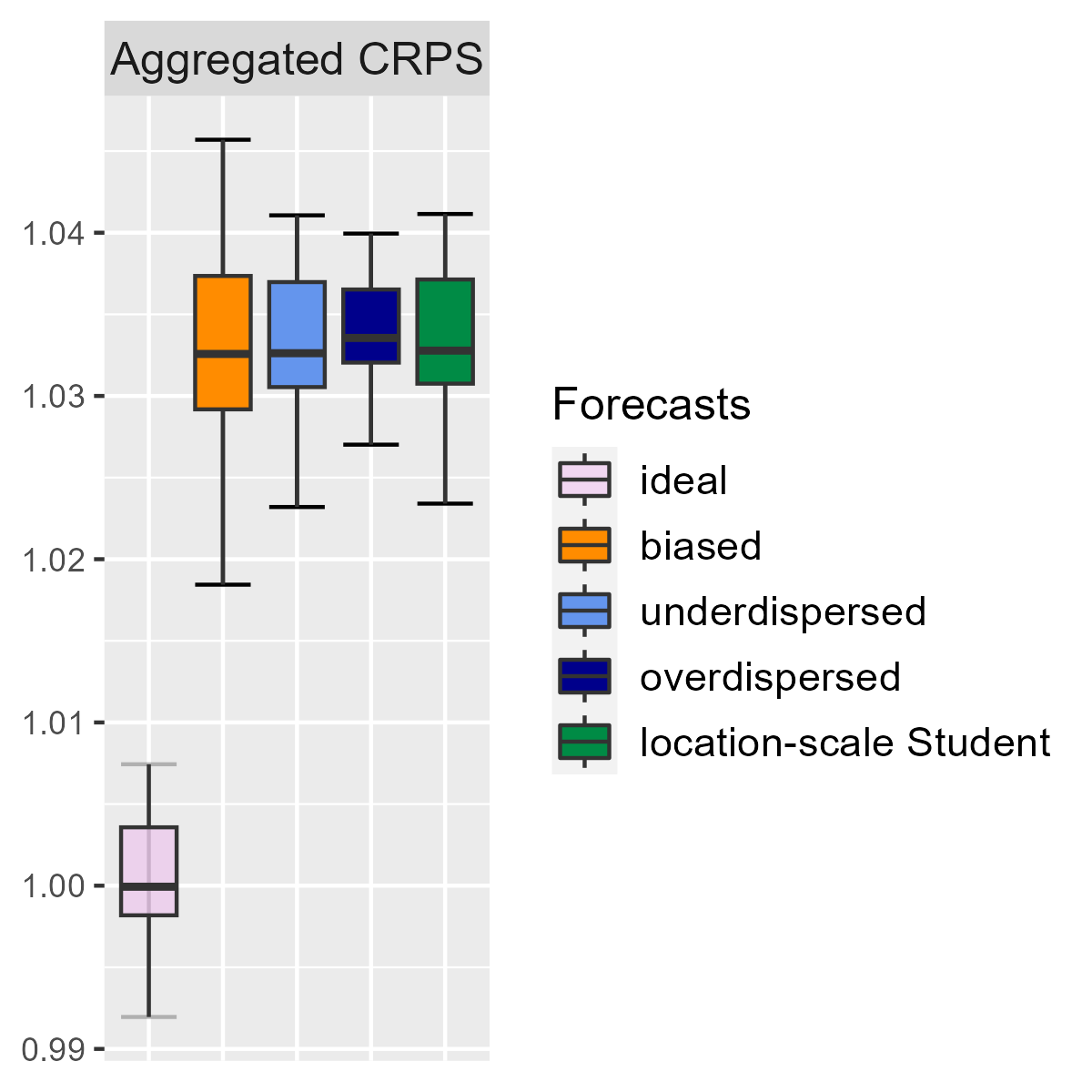} 
            \caption{Aggregated CRPS} 
            \label{fig:marginals_crps} 
            \vspace{4ex}
        \end{flushleft}
    \end{subfigure}
    \begin{subfigure}[c]{.7\textwidth}
        \begin{flushright}
            \includegraphics[scale=.57]{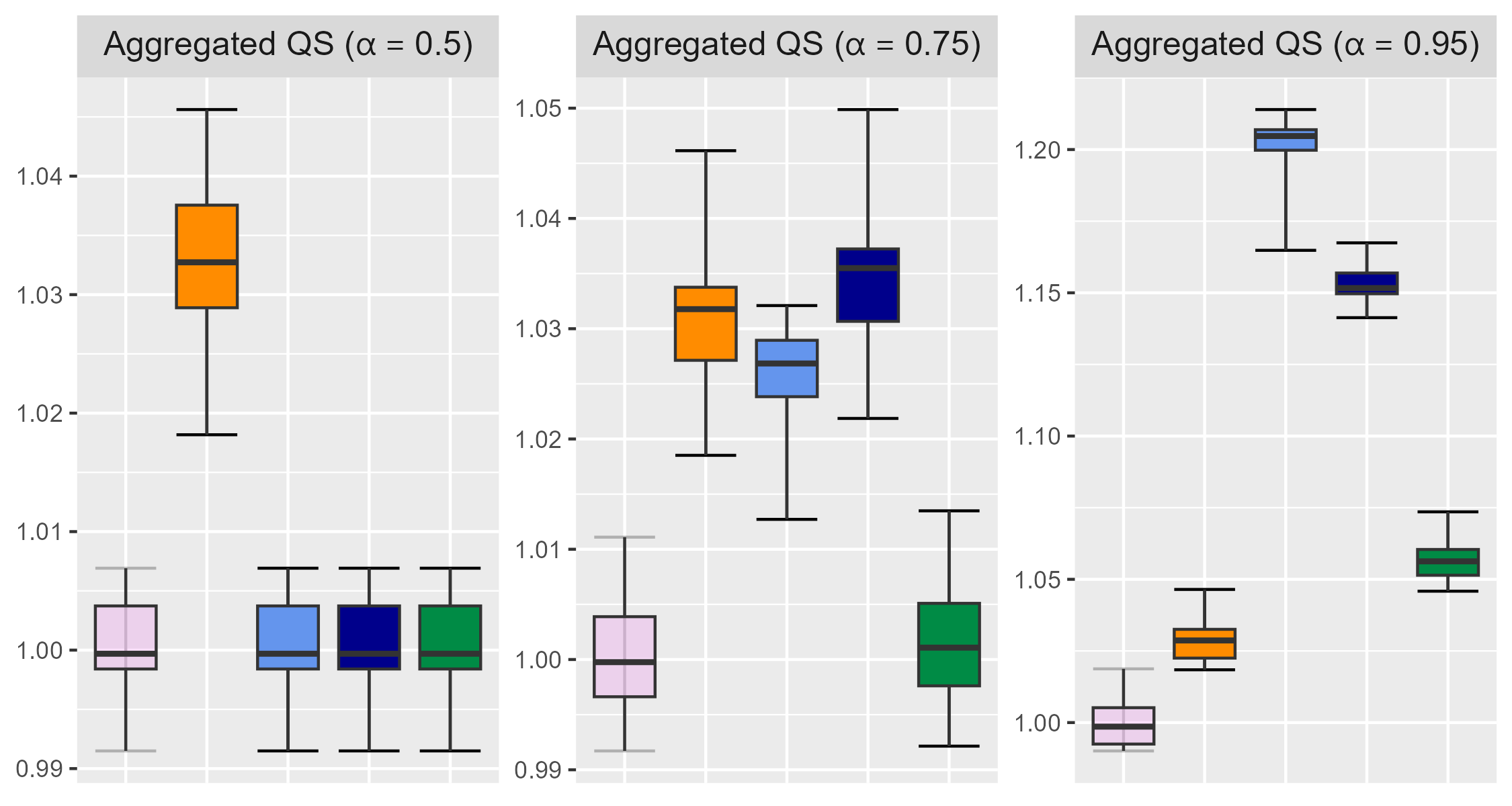} 
            \caption{Aggregated QS} 
        \label{fig:marginals_qs}
        \vspace{4ex}
        \end{flushright}
    \end{subfigure}\\
  
    \begin{subfigure}[c]{0.52\textwidth}
        \begin{flushleft} 
            \includegraphics[scale=.57]{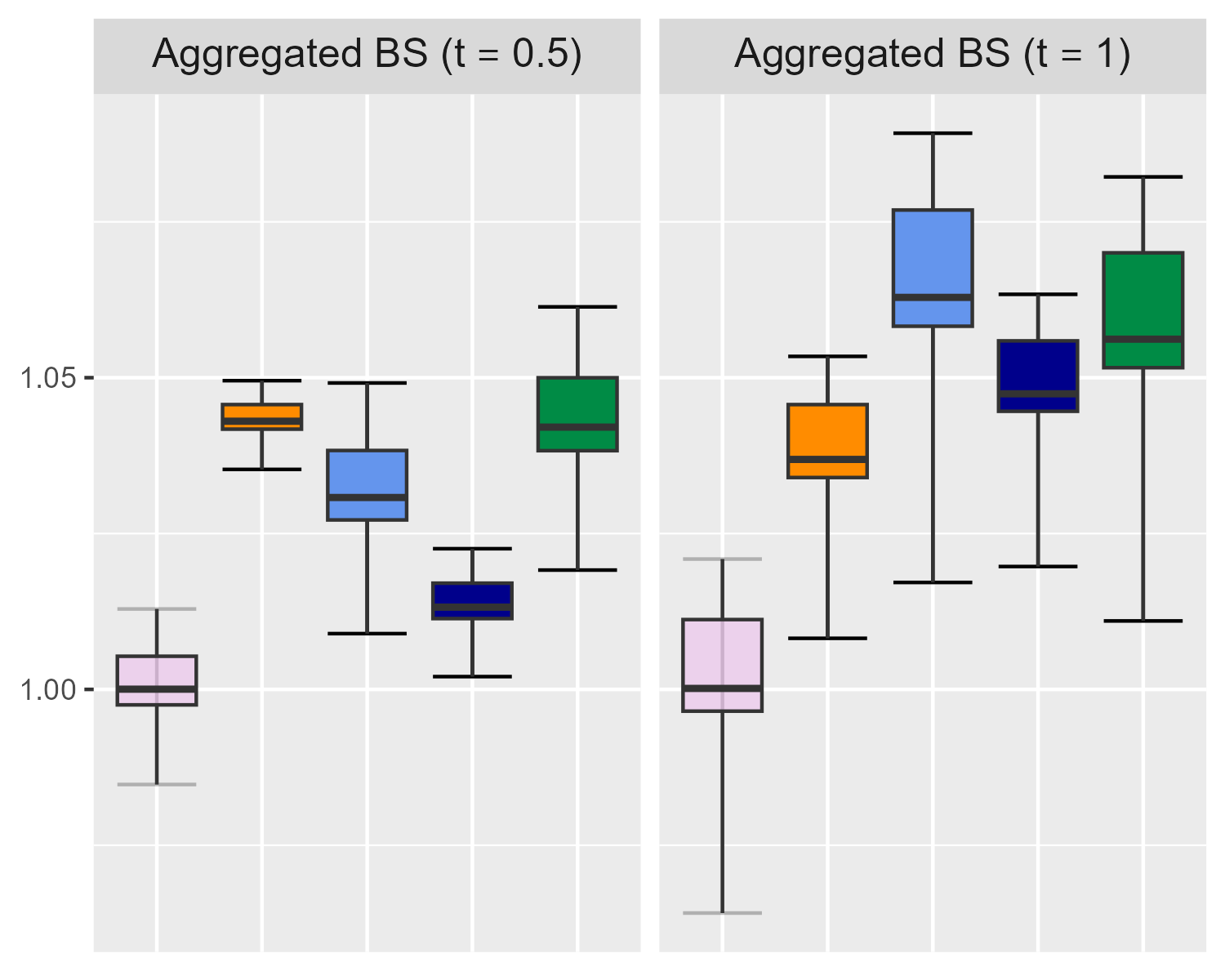} 
            \caption{Aggregated BS} 
            \label{fig:marginals_bs}
        \end{flushleft}
    \end{subfigure}%
    \begin{subfigure}[c]{.48\textwidth}
        \begin{flushright}
            \includegraphics[scale=.57]{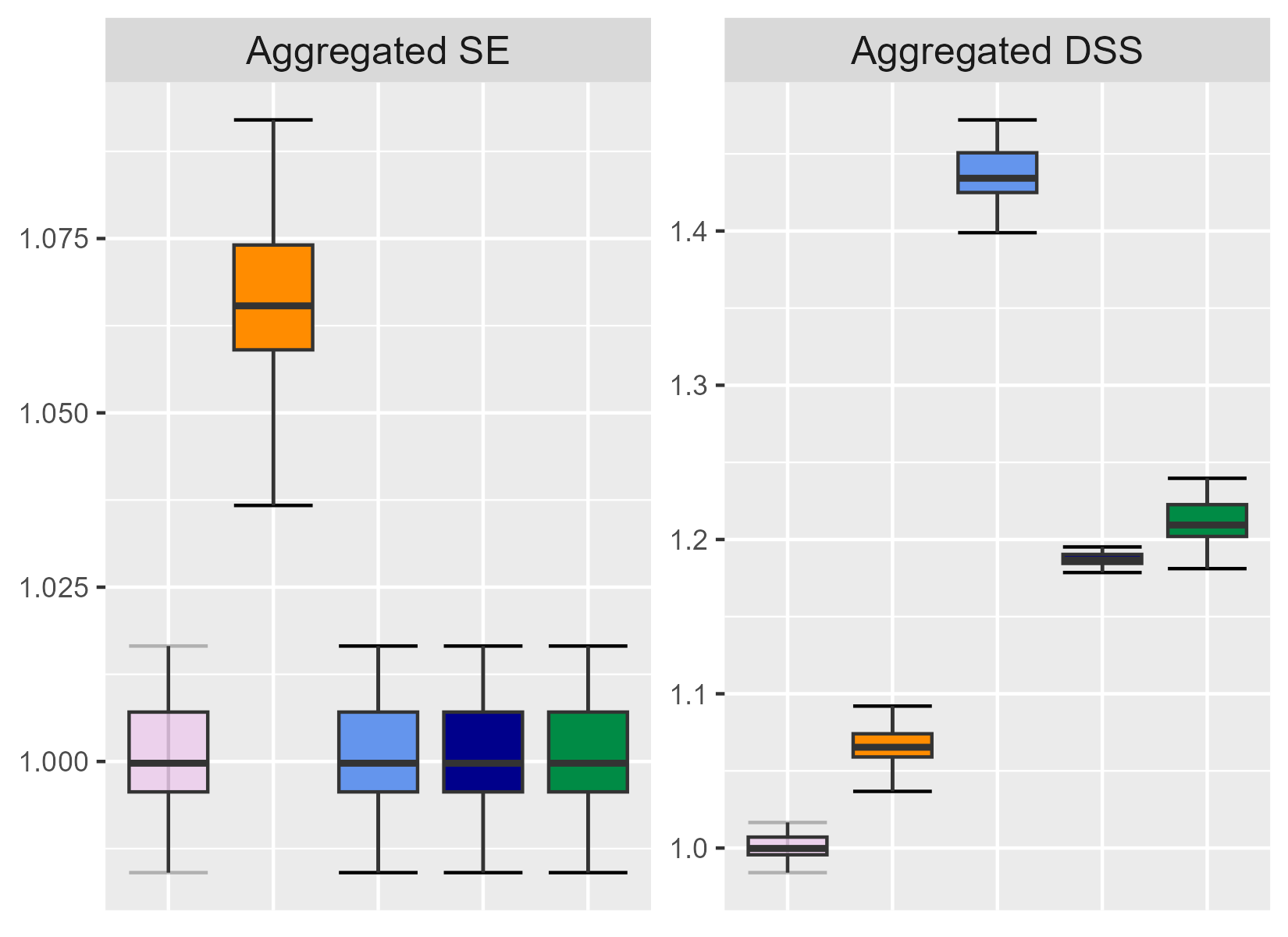} 
            \caption{Aggregated DSS and SE} 
            \label{fig:marginals_moments}
        \end{flushright}
    \end{subfigure}
  \caption{Expectation of aggregated univariate scoring rules: (a) the CRPS, (b) the quantile score, (c) the Brier score, and (d) the squared error and the Dawid-Sebastiani score, for the ideal forecast (light violet), a biased forecast (orange), an under-dispersed forecast (lighter blue), an over-dispersed forecast (darker blue) and a local-scale Student forecast (green). More details are available in the main text.}
  \label{fig:marginals} 
\end{figure}

This first numerical experiment focuses on the prediction of the 1-dimensional marginal distributions and the aggregation of univariate scoring rules. For simplicity, we consider only stationary random fields so that the 1-marginal distribution is the same at all grid points. Although similar conclusions could be drawn from an univariate framework (i.e., with independent 1-dimensional rather than spatial observations), this example aims to clarify the notion of interpretability and presents notions that will be reused in the following examples. The verification of marginals, along with other simple quantities, is usually one of the first steps of any multivariate forecast verification process. 

Observations follow the model of \eqref{eq:obs_simu} and multiple competing forecasts are considered:
\begin{itemize}
    \item[-] the \textit{ideal forecast} is the Gaussian distribution generating observations and is used as a reference;
    \item[-] the \textit{biased forecast} is a Gaussian predictive distribution with the same covariance structure as the observation but a different mean $\bbE[F_\mathrm{bias}(s)]=c=0.255$;
    \item[-]  the \textit{overdispersed forecast} and the \textit{underdispersed forecast} are Gaussian predictive distributions from the same model as the observations except for an overestimation ($\sigma=1.4$) and an underestimation ($\sigma=2/3$) of the variance  respectively; 
    \item[-] the \textit{location-scale Student forecast} is used where the marginals follow location-scale Student-$t$ distributions with parameters $\mu=0$, $df=5$, and $\tau$ is such that the standard deviation is $0.745$ and the covariance structure the same as in~\eqref{eq:obs_simu}.
\end{itemize}

In order to compare the predictive performance of forecasts, we use scoring rules constructed by aggregating univariate scoring rules. Here, the aggregation is done with uniform weights since there is no prior knowledge on the locations. The univariate scoring rules considered are the continuous ranked probability score (CRPS), the Brier score (BS), the quantile score (QS), the squared error (SE) and the Dawid-Sebastiani score (DSS). Figure~\ref{fig:marginals_crps} compares five different forecasts based on their expected CRPS. It can be seen that all forecasts except for the ideal one have similar expected values and no sub-efficient forecast is significantly better than the others. In order to gain more insight into the predictive performance of the forecast, it is necessary to use other scoring rules. In practice, the distribution is unknown; thus, it is impossible to know if a forecast is efficient; it is only possible to provide a ranking linked to the \textit{closeness} of the forecast with respect to the observations. The definition of closeness depends on the scoring rule used: for example, the CRPS defines closeness in terms of the integrated quadratic distance between the two cumulative distribution functions (see, e.g., \citealt{Thorarinsdottir2018}).

If the quantity of interest is the value of a quantile of a certain level $\alpha$, the aggregated QS is an appropriate scoring rule. Figure~\ref{fig:marginals_qs} shows the expected aggregated QS for three different levels $\alpha$ : $\alpha=0.5$, $\alpha=.75$ and $\alpha=0.95$. $\alpha=0.5$ is associated with the prediction of the median and, since all the forecasts are symmetric and only the biased forecast is not centered on zero, the other forecasts are equally the best and efficient forecasts. If the third quartile is of interest ($\alpha=0.75$), the location-scale Student forecast appears as significantly the best (among the non-ideal). For the higher level of $\alpha=0.95$, the biased forecast is significantly the best since its bias error seems to be compensated by its correct prediction of the variance. Depending on the level of interest, the best forecast varies; the only forecast that would appear to be the best regardless of the level $\alpha$ is the ideal forecast, as implied by \eqref{eq:CRPS_qs}.

If a quantity of interest is the exceedance of a threshold $t$ at each location, then the aggregated BS is an interesting scoring rule. Figure~\ref{fig:marginals_bs} shows the expectation of aggregated BS for the different forecasts and for two different thresholds ($t=0.5$ and $t=1$). Among the non-ideal forecasts, there seems to be a clearer ranking than for the CRPS. The overdispersed forecast is significantly the best regarding the prediction of the exceedance of the threshold $t=0.5$ and the biased forecast is significantly the best regarding the exceedance of $t=1$. As for the aggregated quantile score, the best forecast depends on the threshold $t$ considered and the only forecast that is the best regardless of the threshold $t$ is the ideal one (see Eq.~\eqref{eq:CRPS_bs}).

If the moments are of interest, the aggregated SE discriminates the first moment (i.e., the mean) and the aggregated DSS discriminates the first two moments (i.e., the mean and the variance). Figure~\ref{fig:marginals_moments} presents the expected values of these scoring rules for the different forecasts considered in this example. The aggregated SEs of all forecasts, except the biased forecast, are equal since they have the same (correct) marginal means. The aggregated DSS presents the biased forecast as significantly the best one (among non-ideal). This is caused by the combined discrimination of the first two moments of the Dawid-Sebastiani score (see Eq.~\eqref{eq:dss_univariate} and Appendix~\ref{appendix:expected_scores}).

\subsection{Multivariate scores over patches}

This second numerical experiment focuses on the prediction of the dependence structure. Observations are sampled from the model of Eq.~\eqref{eq:obs_simu} and we compare forecasts that differ only in their dependence structure through misspecification of the range parameter $\lambda$ and the smoothness parameter $\beta$:
\begin{itemize}
    \item[-] the \textit{ideal forecast} is the Gaussian distribution generating the observations;
    \item[-]  the \textit{small-range forecast} and the \textit{large-range forecast} are Gaussian predictive distributions from the same model \eqref{eq:obs_simu} as the observations except for an underestimation ($\lambda=1$) and an overestimation ($\lambda=5$), respectively, of the range; 
    \item[-]  the \textit{under-smooth forecast} and the \textit{over-smooth forecast} are Gaussian predictive distributions from the same model as the observations except for an underestimation ($\beta=0.5$) and an overestimation ($\beta=2$), respectively, of the smoothness.
\end{itemize}

Since the forecasts differ only in their dependence structure, scoring rules acting on the 1-dimensional marginals would not be able to distinguish the ideal forecast from the others. We use the variogram score (VS) as a reference since it is known to discriminate misspecification of the dependence structure. We introduce the patched energy score, which results from the aggregation of the ES (with $\alpha=1)$ over patches, defined as
\[
    \ES_{\calP,\bm{w}_\calP}(F,\bmy) = \sum_{P\in\calP} w_P \ES_1(F_P,\bmy_P),
\]
where $\calP$ is an ensemble of spatial patches, $w_P$ is the weight associated with a patch $P\in\calP$ and $F_P$ is the marginal of $F$ over the patch $P$. In order to make the scoring more interpretable, only square patches of a given size $s$ are considered and the weights $w_P$ are uniform ($w_P=1/|\calP|$). Moreover, we consider the aggregated CRPS and the ES since they are limiting cases of the patched ES for $1\times1$ patches and a single patch over the whole domain $\calD$, respectively. Additionally, we proposed the $p$-variation score ($p$VS), which is based on the $p$-variation transformation to focus on the discrimination of the regularity of the random fields,
\[
    T_{p-var,\bms}(\bmX) = |\bmX_{\bms+(1,1)}-\bmX_{\bms+(1,0)}-\bmX_{\bms+(0,1)}+\bmX_{\bms}|^p
\]

\begin{align*}
    p\mathrm{VS}(F,\bmy) &= \sum_{\bms\in\calD^\ast} w_{\bms} \SE_{T_{p-var,\bms}}(F,\bmy);\\
    &= \sum_{\bms\in\calD^\ast} w_{\bms} (\bbE_F[T_{p-var,\bms}(\bmX)]-T_{p-var,\bms}(\bmy))^2,
\end{align*}
where $\calD^\ast$ is the domain $\calD$ restricted to grid points such that $T_{p-var,\bms}$ is defined (i.e., $\calD^\ast=\{1,\ldots,19\}\times\{1,\ldots,19\}$). Note that in the literature on fractional random fields, the $p$-variation is an important characteristic used to characterize the roughness of a random field and is commonly used for estimation purposes, see \cite{Benassi2004}, \cite{BasseO’Connor2021} and the references therein.\\

\begin{figure}[H]
    \centering
    \begin{subfigure}[c]{.5\textwidth}
        \centering
        \includegraphics[scale=.57]{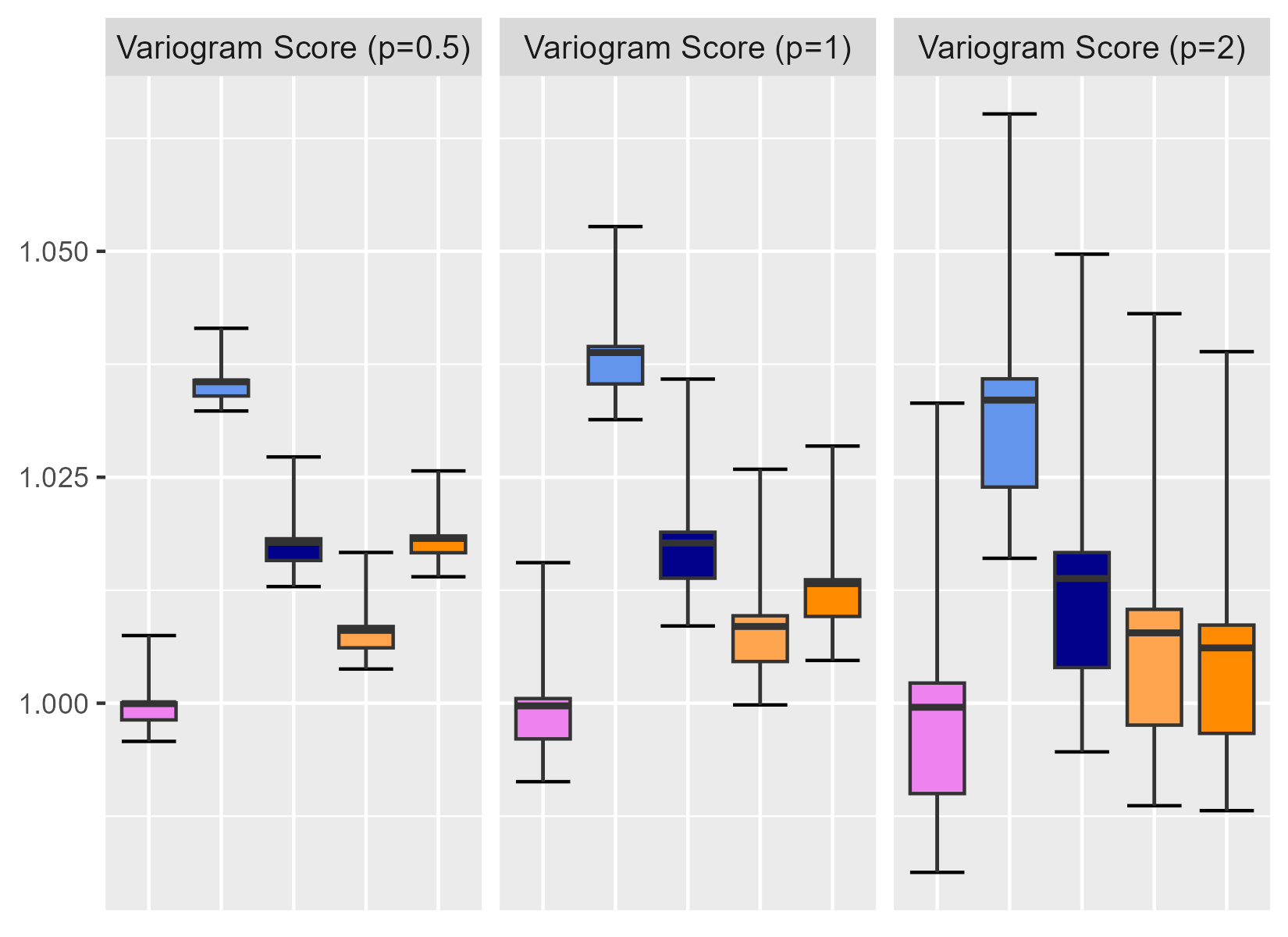}
        \caption{Variogram score}\label{fig:patches_vs}
    \end{subfigure}%
    \begin{subfigure}[c]{.5\textwidth}
        \centering
        \includegraphics[scale=.57]{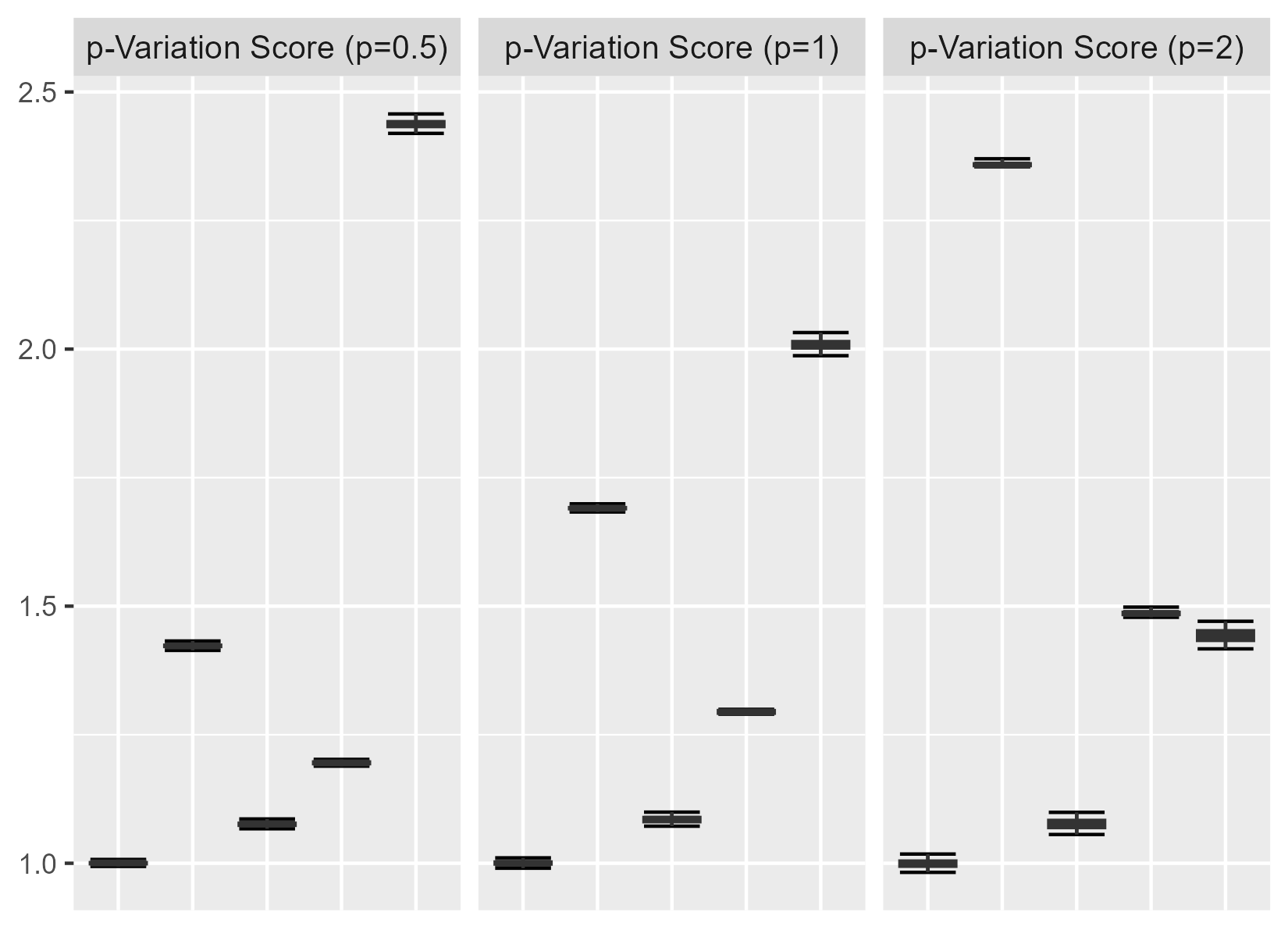}
        \caption{$p$-Variation score}\label{fig:patches_pvs}
    \end{subfigure}
    
    \begin{subfigure}{\textwidth}
        \centering
        \includegraphics[scale=.57]{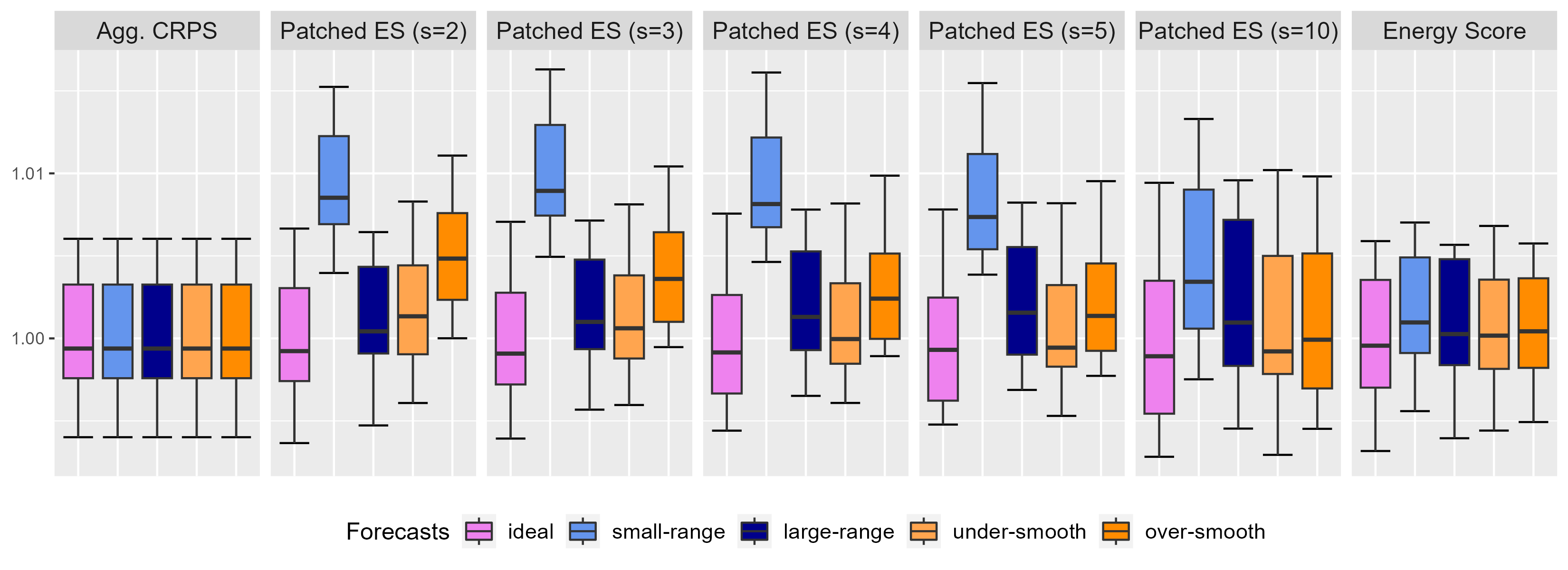}
        \caption{Aggregated CRPS, patched ESs and ES}\label{fig:patches_es}
    \end{subfigure}
    
    \caption{Expectation of scoring rules focused the dependence structure: (a) the variogram score, (b) the $p$-variation score and (c) the patched energy score (and its limiting cases: the aggregated CRPS and the energy score), for the ideal forecast (violet), the small-range forecast (lighter blue), the large-range forecast (darker blue), the under-smooth forecast (lighter orange), and the over-smooth forecast (darker orange). More details are available in the main text.}
    \label{fig:patches}
\end{figure}

In Figure~\ref{fig:patches}, the ES and the patched ES were computed using samples from the forecasts since closed expressions could not be derived. However, closed formulas for the VS and the $p$VS were derived and are available in Appendix~\ref{appendix:sr_simu}. As already shown in \cite{Scheuerer2015}, the VS is able to significantly discriminate misspecification of the dependence structure induced by the range parameter $\lambda$ (see Fig.~\ref{fig:patches_vs}). Smaller orders of $p$ (such as $p=0.5$) appear as more informative than higher ones. Moreover, it is able to discriminate misspecification induced by the smoothness parameter $\beta$ (significantly for all orders $p$ considered) even if it is less marked than for the misspecification of the range $\lambda$.

Figure~\ref{fig:patches_pvs} compares the forecasts using the $p$-variation score with $p\in\{0.5,1,2\}$. Note that the forecasts are provided in the same order as in the other sub-figures. The $p$VS is able to (significantly) discriminate all four sub-efficient forecasts from the ideal forecast at all order $p$. In the cases considered, the $p$VS has a stronger discriminating ability than the VS; in particular, for misspecification of the smoothness parameter $\beta$. The overall improvement in the discrimination ability of the $p$VS compared to the VS is due to the fact that it only considers local pair interactions between grid points; which in the experimental setup considered greatly improves the signal-to-noise ratio compared to the VS. For example, it would be incapable of differentiating two forecasts that only differ in their longer-range dependence structure, where the VS should discriminate the two forecasts. 

Figure~\ref{fig:patches_es} shows that the patched ESs have a better discrimination ability than the ES. As expected by the clear analogy between the variogram score weights and the selection of valid patches, focusing on smaller patches improves the signal-to-noise ratio. For all patch size $s$ considered, the patched ES significantly discriminates the ideal forecast from the others. Whereas the ES does not significantly discriminate the misspecification of smoothness of the under-smooth and over-smooth forecasts. Nonetheless, the patched ES remains less sensitive than the VS to misspecifications in the dependence structure through the range parameter $\lambda$ or the smoothness parameter $\beta$.\\

The VS relies on the aggregation and transformation principles and is able to discriminate the dependence structure. Similarly, the $p$VS is able to discriminate misspecifications of the dependence structure. Being based on more local transformations (i.e., $p$-variation transformation instead of variogram transformation), it has a greater discrimination ability than the VS in this experimental setup. In addition to this known application of the aggregation and transformation principles, it has been shown that multivariate transformations can be used to obtain patched scores that, in the case of the ES, lead to an improvement in the signal-to-noise ratio with respect to the original scoring rule.

\subsection{Anisotropy}

In this example, we focus on the anisotropy of the dependence structure. We introduce geometric anisotropy in observations and forecasts via the covariance function in the following way

\[
    \mathrm{cov}(G(s),G(s')) = \exp\left(-\left(\frac{\lVert s-s'\rVert_A}{\lambda_0}\right)\right)
\]
with $\lVert s-s'\rVert_A=(s-s')^T A (s-s')$. The matrix $A$ has the following form :
\[
    A = \begin{bmatrix}
            \cos\theta & -\sin\theta\\
            \rho\sin\theta & \rho\cos\theta
        \end{bmatrix}
\]
with $\theta\in[-\pi/2,\pi/2]$ the direction of the anisotropy and $\rho$ the ratio between the axes. \\

The observations follow the anisotropic version of the model in Eq.~\eqref{eq:obs_simu} where the covariance function presents the geometric anisotropy introduced above with  $\lambda_0=3$ (as previously) and $\rho_0=2$ and $\theta_0=\pi/4$. Multiple forecasts are considered that only differ in their prediction of the anisotropy in the model:
\begin{itemize}
    \item[-]  the \textit{ideal forecast} has the same distribution as the observations and is used as a reference;
    \item[-] the \textit{small-angle forecast} and the \textit{large-angle forecast} have a correct ratio $\rho$ but an under- and over-estimation of the angle, respectively (i.e., $\theta_{\mathrm{small}}=0$ and $\theta_{\mathrm{large}}=\pi/2$);
    \item[-] the \textit{isotropic forecast} and the \textit{over-anisotropic forecast}  have a ratio $\rho=1$ and $\rho=3$, respectively, but a correct angle $\theta$.
\end{itemize}

\begin{figure}[ht]
    \centering
    \begin{subfigure}{\textwidth}
        \centering
      \includegraphics[scale=.57]{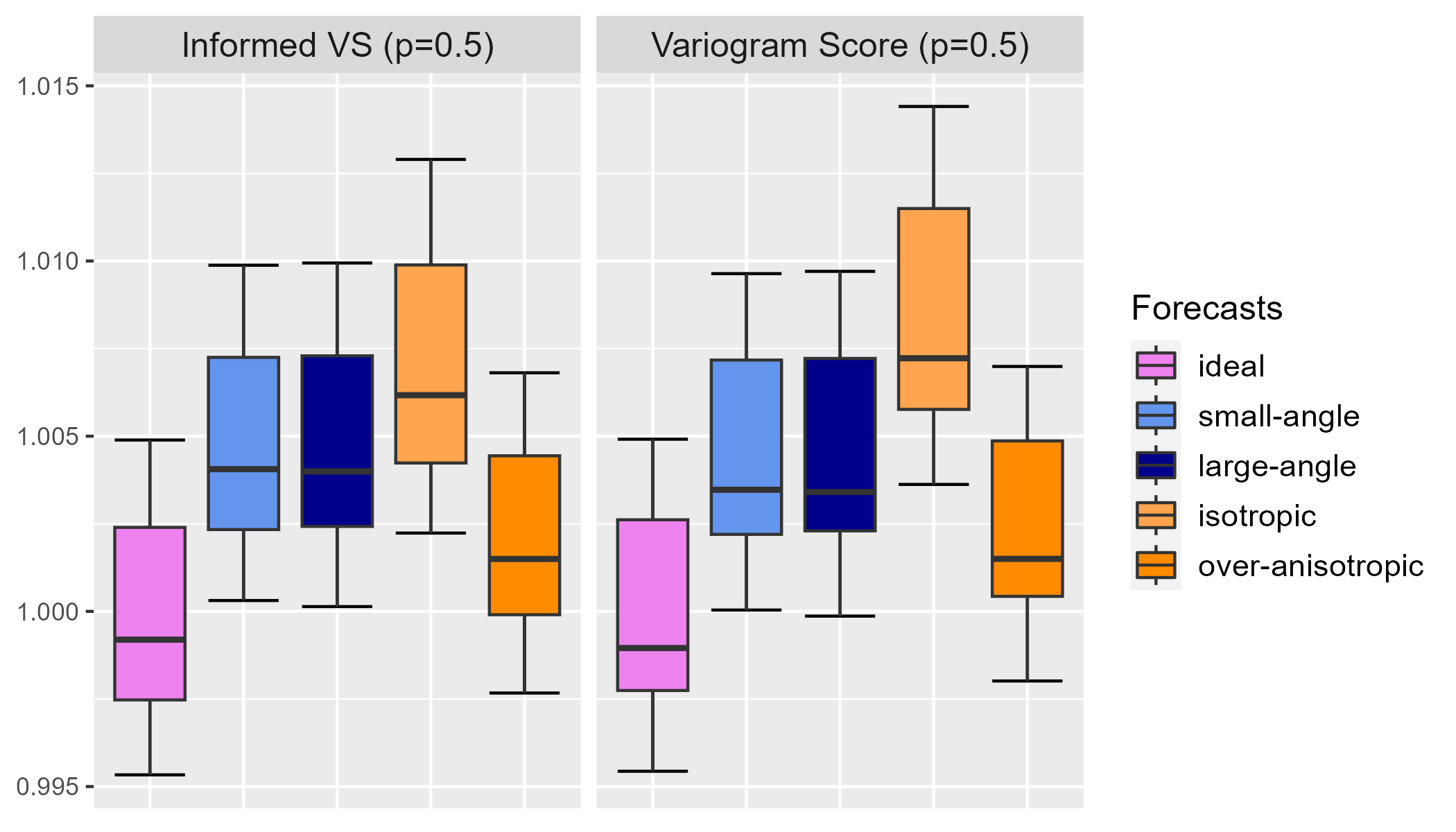}
      \caption{Variogram score}\label{fig:anisotropy_vs}
    \end{subfigure}
    
    \begin{subfigure}{\textwidth}
      \includegraphics[scale=.57]{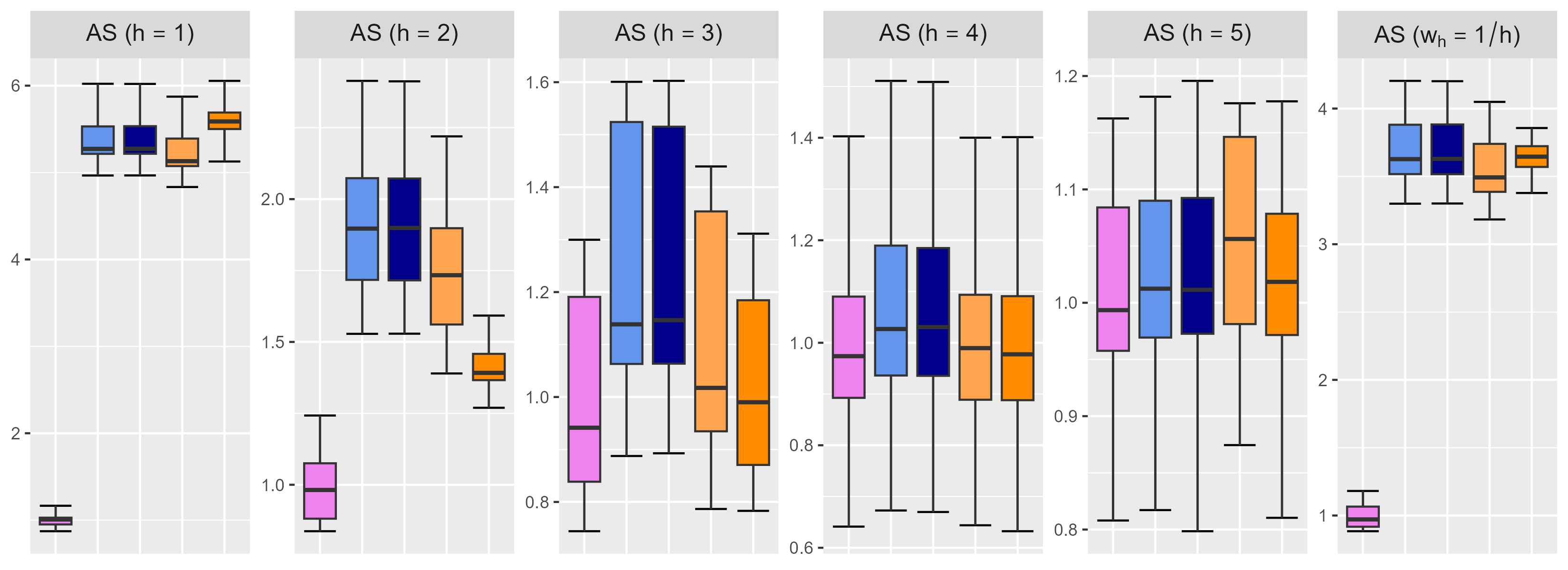}
      \caption{Anisotropic score for different scales $h$ and aggregated across scales ($w_h=1/h$)}\label{fig:anisotropy_as_a2}
    \end{subfigure}
    
    \caption{Expectation of interpretable proper scoring rules focused the dependence structure: (a) the variogram score and (b) the anisotropic score, for the ideal forecast (violet), the small-angle forecast (lighter blue), the large-angle forecast (darker blue), the isotropic forecast (lighter orange) and the over-anisotropic forecast (darker orange). More details are available in the main text.}
    \label{fig:anisotropy}
\end{figure}

Since these forecasts differ only in the anisotropy of their dependence structure, scoring rules not suited to discriminate the dependence structure would not be able to differentiate them. We compare two proper scoring rules: the variogram score and the anisotropic scoring rule. The variogram score is studied in two different settings: one where the weights are proportional to the inverse of the distance and one where the weights are proportional to the inverse of the anisotropic distance $\lVert\cdot\rVert_A$, which is supposed to be more informed since it is the quantity present in the covariance function. The anisotropic score (AS) is a scoring rule based on the framework introduced in Section~\ref{section:framework} and, in its general form, it is defined as 
\begin{equation}\label{eq:as_wh}
    \AS(F,\bmy) = \sum_{h} w_h \rmS_{T_{\mathrm{iso},h}}(F,\bmy) = \sum_h w_h \rmS(T_{\mathrm{iso},h}(F),T_{\mathrm{iso},h}(\bmy)),
\end{equation}
where $T_{\mathrm{iso},h}$ is a transformation summarizing the anisotropy of a field such as the one introduced in \eqref{eq:T_iso}. Additionally, we use a special case of this scoring rule where we do not aggregate across the scales $h$ and where $\rmS$ is the squared error :
\begin{equation}\label{eq:anisotropic_score}
    S_{T_{\mathrm{iso},h}}(F,\bmy) = \SE(T_{\mathrm{iso},h}(F),T_{\mathrm{iso},h}(\bmy)) = \left(\bbE_{T_{\mathrm{iso},h}(F)}[X]-T_{\mathrm{iso},h}(\bmy)\right)^2.
\end{equation}

We use a transformation similar to the one of \eqref{eq:T_iso} where instead the axes are the first and second bisectors. This leads to the following formula:

\[
    T_{\mathrm{iso},h}(\bmX) = - \cfrac{\big(\gamma_X((h,h))-\gamma_X((-h,h))\big)^2}{\cfrac{2\gamma_X((h,h))^2}{|\calD((h,h))|}+\cfrac{2\gamma_X((-h,h))^2}{|\calD((-h,h))|}} .
\]

The choice of this transformation instead of the transformation based on the anisotropy along the abscissa and ordinate is motivated by the fact that it leads to a clearer differentiation of the forecasts (not shown).\\

Figure~\ref{fig:anisotropy_vs} presents the variogram score of order $p=0.5$ in its two variants. Both the standard VS and the informed VS are able to significantly discriminate the ideal forecast from the other forecasts but they have a weak sensitivity to misspecification of the geometric anisotropy. Even though the informed VS is supposed to increase the signal-to-noise ratio compared to the standard VS; it is not more sensitive to misspecifications in the experimental setup considered. Other orders of variograms were tested and worsened the discrimination ability of both standard and informed VS (not shown).

Figure~\ref{fig:anisotropy_as_a2} shows the AS \eqref{eq:anisotropic_score} with scales $1\leq h\leq 5$ for the different forecasts and the aggregated AS \eqref{eq:as_wh}, where the scales are aggregated with weights $w_h=1/h$. The anisotropic scores were computed using samples drawn from the forecasts; this explains why the ideal forecast may appear sub-efficient for some values of $h$ (e.g., $h=4$). As aimed by its construction, the AS is able to significantly distinguish the correct anisotropy behavior in the dependence structure for values of $h$ up to $h=3$ included. For $h=4$, the AS does not significantly discriminate the isotropic forecast and the over-anisotropic forecast from the ideal one. The fact that $h=1$ is the most sensitive to misspecifications is probably caused by the fact that the strength of the dependence structure decays with the distance (i.e., with $h$). This shows that the hyperparameter $h$ plays an important role in having an informative AS (as do the weights and the order $p$ for the variogram score). For $h=2$ in particular, it can be seen that the AS is not sensitive to the misspecification of the ratio $\rho$ and the angle $\theta$ in the same manner. This depends on the degree of misspecification but also on the hyperparameters of the AS. The aggregated AS allows us to avoid the selection of a scale $h$ while maintaining the discrimination ability of the lower values of $h$.\\

The anisotropic score is an interpretable scoring rule targeting the anisotropy of the dependence structure. However, it has the limitation of introducing hyperparameters in the form of the scale $h$ and the axes along which the anisotropy is measured. Aggregation across scales and axes can circumvent the selection of these hyperparameters; however, a clever choice of weights will be required to maintain the signal-to-noise ratio.

\subsection{Double-penalty effect}

In this example, we illustrate in a simple setting how scoring rules over patches can be robust to the double-penalty effect (see Section~\ref{subsection:spatial_verif}). The double-penalty effect is introduced in the form of forecasts that deviate from the ideal forecast by an additive or multiplicative noise term (i.e., nugget effect). The noises are centered uniforms such that the forecasts are correct on average but incorrect over each grid point. 

Observations follow the model of \eqref{eq:obs_simu} with the parameters $\sigma_0=1$, $\lambda_0=3$ and $\beta_0=1$. As per usual the \textit{ideal forecast}, having the same distribution as the observations, is used as a reference. \textit{Additive-noised forecasts} are the first type of forecast introduced to test the sensitivity of scoring rules to the form of the double-penalty effect (presented above). They differ from the ideal forecast through their marginals in the following way :
\[
    F_{\mathrm{add}}(s) = \calN(\epsilon_{s},\sigma_0^2),
\]
where $\epsilon_{s}\sim\mathrm{Unif}([-r,r])$ is a  spatial white noise independent at each location $s\in\calD$. This has an effect on the mean of the marginals at each grid point. Three different noise range values are tested $r\in\{0.1, 0.25, 0.5\}$. Similarly, \textit{multiplicative-noised forecasts} that differ from the ideal forecast through their marginals are introduced :
\[
    F_{\mathrm{mul}}(s) = \calN(0,\sigma^2(1+\eta_{s})^2),
\]
where $\eta_{s}\sim\mathrm{Unif}([-r,r])$ and $s\in\calD$. This has an effect on the variance of the marginals at each grid point and, thus, on the covariance. The same noise range values are tested $r\in\{0.1, 0.25, 0.5\}$.\\

The aggregated CRPS is a naive scoring rule that is sensitive to the double-penalty effect. We propose the aggregated CRPS of spatial mean which is defined as 
\begin{align*}
    \CRPS_{\mathrm{mean}_\calP,\bm{w_\calP}}(F,\bmy) &= \sum_{P\in\calP} w_P \CRPS_{\mathrm{mean}_P}(F,\bmy);\\
    &= \sum_{P\in\calP} w_P \CRPS(\mathrm{mean}_P(F),\mathrm{mean}_P(\bmy)),
\end{align*}
where $\calP$ is an ensemble of spatial patches, $w_P$ is the weight associated with a patch $P\in\calP$ and $\mathrm{mean}_P$ the spatial mean over the patch $P$ \eqref{eq:mean_P}. It is a proper scoring rule, and it has an interpretation similar to the aggregated CRPS, but the forecasts are only evaluated on the performance of their spatial mean. In order to make the scoring more interpretable, only square patches of a given size $s$ are considered and the weights $w_P$ are uniform. The size of the patches $s$ can be determined by multiple factors such as the physics of the problem, the constraints of the verification in the case of models on different scales, or hypotheses on a different behavior below and above the scale of the patch (e.g., independent and identically distributed; \citealt{Taillardat2020}). Note that the aggregated CRPS of spatial mean is equal to the aggregated CRPS when patches of size $s=1$ are considered.

\begin{figure}[!ht]
    \centering
    \begin{subfigure}{\textwidth}
        \centering
        \includegraphics[scale=.57]{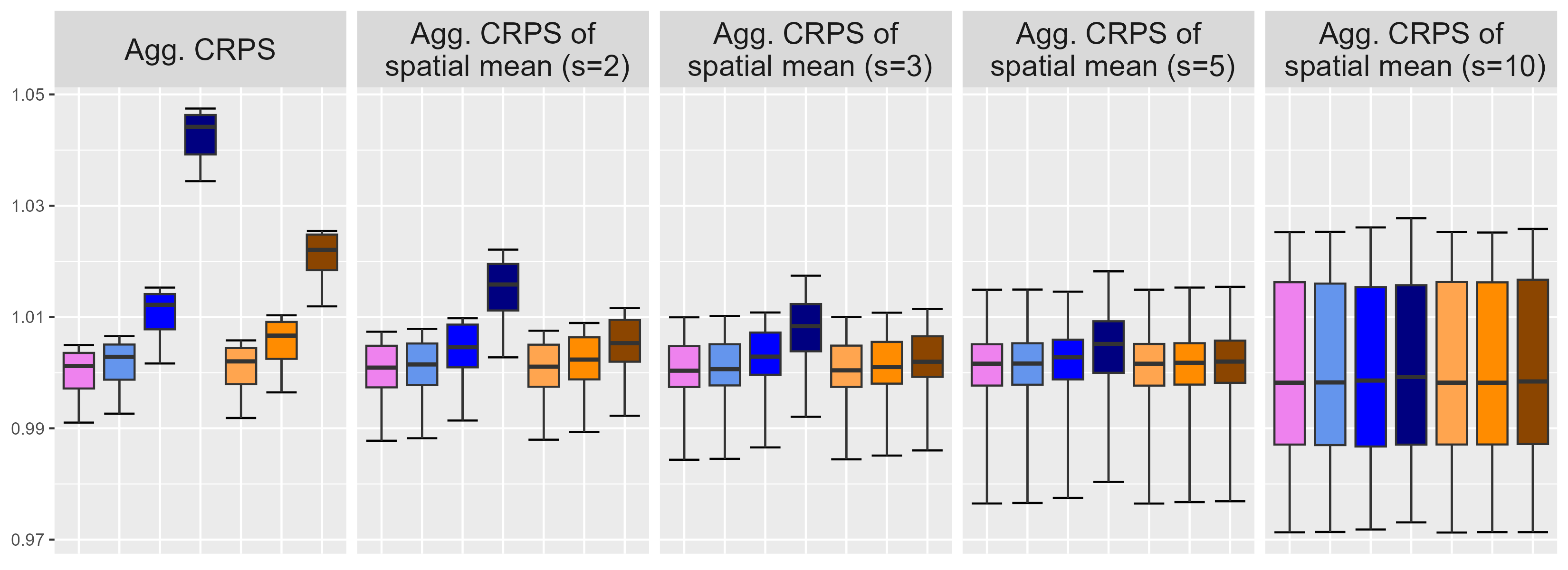}
        \caption{Aggregated CRPS and CRPS of spatial mean}\label{fig:double_penalty_crps}
    \end{subfigure}
    
    \begin{subfigure}{\textwidth}
        \centering
        \includegraphics[scale=.57]{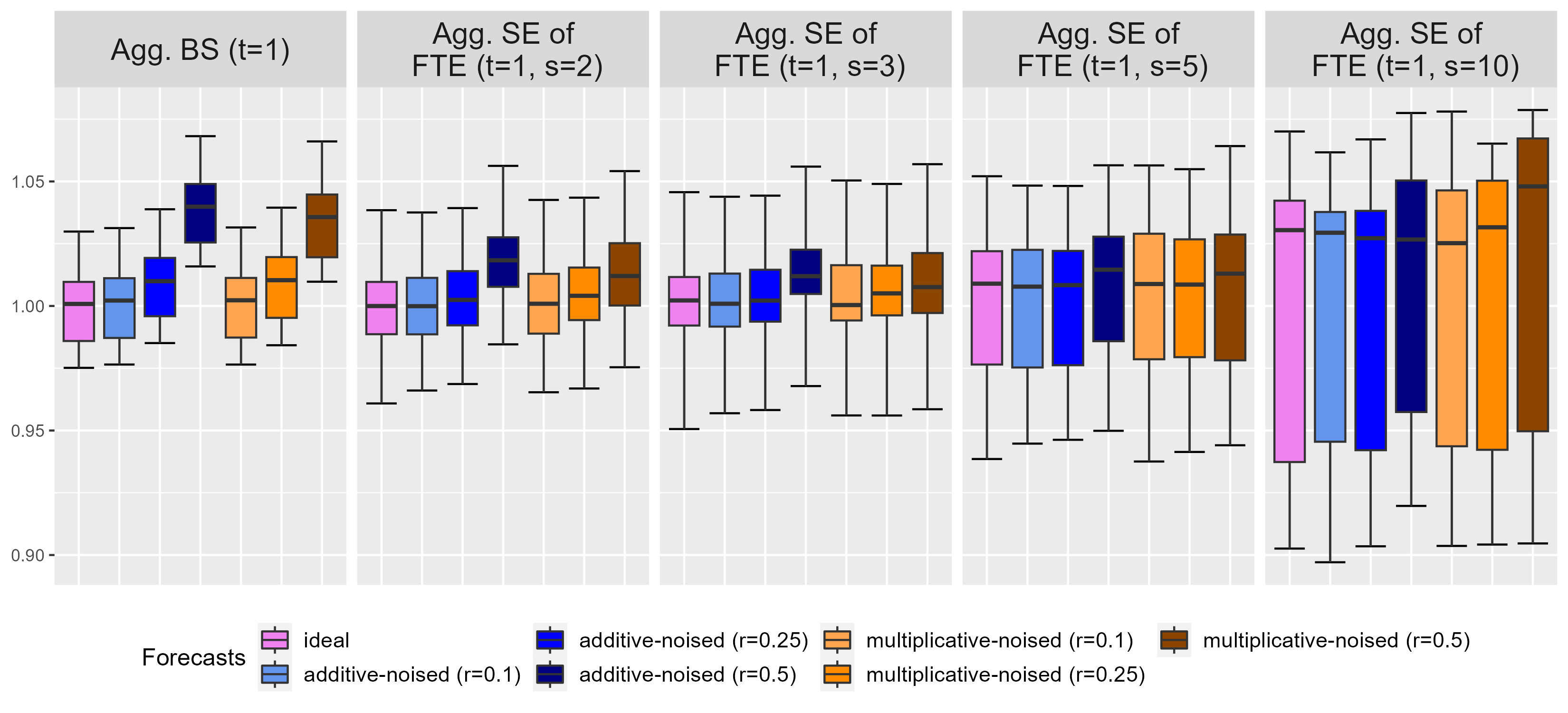}
        \caption{Aggregated BS and SE of FTE}
        \label{fig:double_penalty_bs}
    \end{subfigure}
    
    \caption{Expectation of scoring rules tested on their sensitivity to double-penalty effect : (a) the aggregated CRPS and the aggregated CRPS of spatial mean, and (b) the aggregated Brier score and the aggregated squared error of fraction of threshold exceedances, for the ideal forecast (violet), the additive-noised forecasts (shades of blue), and the multiplicative-noised forecasts (shades of orange). For the noised forecasts, darker colors correspond to larger values of the range $r\in\{0.1,\ 0.25,\ 0.5\}$. More details are available in the main text.}
    \label{fig:double_penalty}
\end{figure}

If a quantity of interest is the exceedance of a threshold $t$, the scoring rule associated with that is the Brier score \eqref{eq:BS}. We compare the aggregated BS with its counterpart over patches: the aggregated SE of the FTE. It is defined as
\begin{align*}
    \SE_{\mathrm{FTE}_{\calP,t},\bm{w_\calP}}(F,\bmy) &= \sum_{P\in\calP} w_P \SE_{\mathrm{FTE}_{P,t}}(F,\bmy);\\
    &= \sum_{P\in\calP} w_P \SE\big(\mathrm{FTE}_{P,t}(F),\mathrm{FTE}_{P,t}(\bmy)\big)\\
    &= \sum_{P\in\calP} w_P \big(\bbE_F[\mathrm{FTE}_{P,t}(X)]-\mathrm{FTE}_{P,t}(\bmy)\big)^2
\end{align*}
where $\calP$ is an ensemble of spatial patches, $w_P$ is the weight associated with a patch $P\in\calP$ and $\mathrm{FTE}_{P,t}$ the fraction of threshold exceedance over the patch $P$ and for the threshold $t$ \eqref{eq:fte}. This scoring rule is proper and focuses on the prediction of the exceedance of a threshold $t$ via the fraction of locations over a patch $P$ exceeding said threshold. The resemblance with the Brier score is clear and the aggregated SE of FTE becomes the aggregated BS when patches of size $s=1$ are considered.\\

In Figure~\ref{fig:double_penalty}, the values of the aggregated SE of FTE have been obtained by sampling the forecasts' distribution. Figure~\ref{fig:double_penalty_crps} compares the aggregated CRPS and the aggregated CRPS of spatial mean for different patch size $s$. For all the scoring rules, we observe an increase in the expected value with the increase of the range of the noise $r$. As expected, the aggregated CRPS is very sensitive to noise in the mean or the variance and, thus, is prone to the double-penalty effect. The aggregated CRPS of spatial mean is less sensitive to noise on the mean or the variance. Moreover, different patch sizes allow us to select the spatial scale below which we want to avoid a double penalty. Given that the distribution of the noise is fixed in this simulation (i.e., uniform), patch size is related to the level of random fluctuations (i.e., the range $r$) tolerated by the scoring rule before significant discrimination with respect to the ideal forecast. It is worth noting that the range $r$ of the noise leads to a stronger increase in the values of these CRPS-related scoring rules when the noise is on the mean rather than on the variance.

Figure~\ref{fig:double_penalty_bs} compares the aggregated BS and the aggregated squared error of fraction of threshold exceedances. For simplicity, we fix the threshold $t=1$. The aggregated BS is, as expected, sensitive to noise in the mean or the variance, and an increase in the range of the noise leads to an increase in the expected value of the score. The aggregated SE of FTE acts as a natural extension of the aggregated BS to patches and provides scoring rules that are less sensitive to noise on the mean or the variance. The sensitivity evolves differently with the increase of the patch size $s$ compared to the aggregated CRPS of spatial mean since the aggregated SE of FTE measures the effect on the average exceedance over a patch. The range $r$ of the noise apparently leads to a comparable increase in the values of the aggregated SE of FTE when the noise is additive or multiplicative.\\

The use of transformations over patches is similar to neighborhood-based methods in the spatial verification tools framework. Even though avoiding the double-penalty effect is not restricted to tools that do not penalize forecasts below a certain scale, this simulation setup presents a type of test relevant to probability forecasts. The patched-based scoring rules proposed here are not by themselves a sufficient verification tool since they are insensitive to some unrealistic forecast (e.g., if the mean value over the patch is correct but deviations may be as large as possible and lead to unchanged values of the scoring rule). As for any other scoring rule, they should be used with other scoring rules. 

\section{Conclusion}\label{section:conclusion}

Verification of probabilistic forecasts is an essential but complex step of all forecasting procedures. Scoring rules may appear as the perfect tool to compare forecast performance since, when proper, they can simultaneously assess calibration and sharpness. However, propriety, even if strict, does not ensure that a scoring rule is relevant to the problem at hand. With that in mind, we agree with the recommendation of \cite{Scheuerer2015} that "several different scores be always considered before drawing conclusions". This is even more important in a multivariate setting where forecasts are characterized by more complex objects. 

We proposed a framework to construct proper scoring rules in a multivariate setting using aggregation and transformation principles. Aggregation-and-transformation-based scoring rules can improve the conclusions drawn since they enable the verification of specific aspects of the forecast (e.g., anisotropy of the dependence structure). This has been illustrated both using examples from the literature and numerical experiments showcasing different settings. Moreover, we showed that the aggregation and transformation principles can be used to construct scoring rules that are proper, interpretable, and not affected by the double-penalty effect. This could be a starting point to help bridging the gap between the proper scoring rule community and the spatial verification tools community. \\

As the interest for machine learning-based weather forecast is increasing (see, e.g., \citealt{BenBouallegue2024Rise}), multiple approaches have performance comparable to ECMWF deterministic high-resolution forecasts \citep{Keisler2022, Pathak2022, Bi2023, Lam2022, Chen2023}. The natural extension to probabilistic forecast is already developing and enabled by publicly available benchmark datasets such as WeatherBench 2 \citep{Rasp2024}. Aggregation-and-transformation-based methods can help ensure that parameter inference does not hedge certain important aspects of the multivariate probabilistic forecasts.\\

There seems to be a trade-off between discrimination ability and strict propriety. Discrimination ability comes from the ability of scoring rules to differentiate misspecification of certain characteristics. By definition, the expectation of strictly proper scoring rules is minimized when the probabilistic forecast is the true distribution. Nonetheless, it does not guarantee that this global minimum is steep in any misspecification direction. However, interpretable scoring rules can discriminate the misspecification of their target characteristic. Should scoring rules discriminating any misspecification be pursued? Or should interpretable scoring rules discriminating a specific type of misspecification be used instead?

\section*{Acknowledgments}

The authors acknowledge the support of the French Agence Nationale de la Recherche (ANR) under reference ANR-20-CE40-0025-01 (T-REX project) and the Energy-oriented Centre of Excellence II (EoCoE-II), Grant Agreement 824158, funded within the Horizon2020 framework of the European Union. Part of this work was also supported by the ExtremesLearning grant from 80 PRIME CNRS-INSU and this study has received funding from Agence Nationale de la Recherche - France 2030 as part of the PEPR TRACCS program under grant number ANR-22-EXTR-0005 and the ANR EXSTA.\\
Sam Allen is thanked for fruitful discussions during the preparation of this manuscript.

\bibliography{transformation-aggregation_refs.bib}

\appendix

\section{Expected univariate scoring rules}\label{appendix:expected_scores}

\subsection{Squared Error}

For any $F,G\in\calP_2(\bbR)$, the expectation of the squared error \eqref{eq:SE} is :
\begin{equation*}
    \bbE_G[\SE(F,Y)] = (\mu_F-\mu_G)^2 + {\sigma_G}^2,
\end{equation*}
where $\mu_F$ is the mean of the distribution $F$ and $\mu_G$ and ${\sigma_G}^2$ are the mean and the variance of the distribution $G$.

\begin{proof}
    \begin{align*}
        \bbE_G[\SE(F,Y)] &= \bbE_G[(\mu_F-Y)^2]\\
        &= \mu_F^2-2\ \mu_F \bbE_G[Y] + \bbE_G[Y^2]\\
    \end{align*}
    Using the fact that $\bbE[X^2]=\mathrm{Var}(X)+\bbE[X]^2$,
    \begin{align*}
        \bbE_G[\SE(F,Y)] &= \mu_F^2-2\ \mu_F \mu_G  + \sigma_G^2 + \mu_G^2 \\
        &= (\mu_F - \mu_G)^2  + \sigma_G^2
    \end{align*}
\end{proof}

\subsection{Quantile Score}

For any $F,G\in\calP_1(\bbR)$, the expectation of the quantile score of level $\alpha$ \eqref{eq:QS} is :
\begin{align*}
    \bbE_G[\QS_\alpha(F,Y)] &= \int_{-\infty}^{F^{-1}(\alpha)} (F^{-1}(\alpha)-y) G(\rmd y) - \alpha\int_\bbR  (F^{-1}(\alpha)-y) G(\rmd y);\\
    &= \bbE_G[\QS_\alpha(G,Y)] + \left\{(G(F^{-1}(\alpha))-\alpha)(F^{-1}(\alpha)-G^{-1}(\alpha)) - \int_{G^{-1}(\alpha)}^{F^{-1}(\alpha)} (y-G^{-1}(\alpha)) G(\rmd y)\right\}.
\end{align*}

\begin{proof}
    \textit{Inspired by the proof of the propriety of the quantile score in \cite{Friederichs2008}}.\\
    \begin{align*}
        \bbE_G[\QS_\alpha(F,Y)] &= \int_\bbR  (\mathds{1}_{y\leq F^{-1}(\alpha)}-\alpha)(F^{-1}(\alpha)-y) G(\rmd y)\\
        &= \int_{-\infty}^{F^{-1}(\alpha)} (1-\alpha)(F^{-1}(\alpha)-y) G(\rmd y) + \int_{F^{-1}(\alpha)}^{+\infty}  (-\alpha)(F^{-1}(\alpha)-y) G(\rmd y)\\
        &= \int_{-\infty}^{F^{-1}(\alpha)} (F^{-1}(\alpha)-y) G(\rmd y) -\alpha\int_\bbR (F^{-1}(\alpha)-y) G(\rmd y)\\
    \end{align*}
    Then, using $F^{-1}(\alpha)-y = (F^{-1}(\alpha)-G^{-1}(\alpha))+(G^{-1}(\alpha)-y)$,
    \begin{align*}
        \bbE_G[\QS_\alpha(F,Y)] &= \int_{-\infty}^{F^{-1}(\alpha)} (F^{-1}(\alpha)-G^{-1}(\alpha)) G(\rmd y) -\alpha\int_\bbR (F^{-1}(\alpha)-G^{-1}(\alpha)) G(\rmd y)\\ &\ \ \ + \int_{-\infty}^{F^{-1}(\alpha)} (G^{-1}(\alpha)-y) G(\rmd y) -\alpha\int_\bbR (G^{-1}(\alpha)-y) G(\rmd y)\\
        &= (G(F^{-1}(\alpha))-\alpha)(F^{-1}(\alpha)-G^{-1}(\alpha)) \\ &\ \ \ + \int_{-\infty}^{F^{-1}(\alpha)} (G^{-1}(\alpha)-y) G(\rmd y) -\alpha\int_\bbR (G^{-1}(\alpha)-y) G(\rmd y)\\
        &= (G(F^{-1}(\alpha))-\alpha)(F^{-1}(\alpha)-G^{-1}(\alpha)) \\ &\ \ \ + \int_{-\infty}^{G^{-1}(\alpha)} (G^{-1}(\alpha)-y) G(\rmd y) + \int_{G^{-1}(\alpha)}^{F^{-1}(\alpha)} (G^{-1}(\alpha)-y) G(\rmd y) -\alpha\int_\bbR (G^{-1}(\alpha)-y) G(\rmd y)\\
        &= (G(F^{-1}(\alpha))-\alpha)(F^{-1}(\alpha)-G^{-1}(\alpha)) + \bbE_G[\QS_\alpha(G,Y)]) - \int_{G^{-1}(\alpha)}^{F^{-1}(\alpha)} (y-G^{-1}(\alpha)) G(\rmd y)\\
    \end{align*}
\end{proof}

\subsection{Absolute Error}

First of all, for $F\in\calP_1(\bbR)$ and $y\in\bbR$, the absolute error \eqref{eq:AE} is equal to twice the quantile score of level $\alpha=0.5$ :
\begin{equation*}
    \rmAE(F,y) = |\mathrm{med}(F)-y| = 2\ \QS_{0.5}(F,y),
\end{equation*}
where $\mathrm{med}(F)$ is the median of the distribution $F$.

It can be deduced that, for any $F,G\in\calP_1(\bbR)$, the expectation of the absolute error is :
\begin{align*}
    \bbE_G[\rmAE(F,Y)] &= \bbE_G[|\mathrm{med}(F)-Y|];\\
                       &= 2\ \int_{-\infty}^{\mathrm{med}(F)} (\mathrm{med}(F)-y) G(\rmd y) - 2\alpha\int_\bbR  (\mathrm{med}(F)-y) G(\rmd y);\\
                       &= \bbE_G[\rmAE(G,Y)] + 2\left\{(G(\mathrm{med}(F))-\alpha)(\mathrm{med}(F)-\mathrm{med}(G)) - \int_{\mathrm{med}(G)}^{\mathrm{med}(F)} (y-\mathrm{med}(G)) G(\rmd y)\right\}.
\end{align*}

\subsection{Brier score}

For any $F,G\in\calP(\bbR)$, the expectation of the Brier score \eqref{eq:BS} is :
\begin{align*}
    \bbE_G[\BS_t(F,Y)] = (F(t)-G(t))^2 + G(t)(1-G(t)).
\end{align*}

\begin{proof}
    \begin{align*}
    \bbE_G[\BS_t(F,Y)] &= \bbE_G[(F(t)-\mathds{1}_{Y\leq t})^2]\\
                     &= F(t)^2-2 F(t) \bbE_G[\mathds{1}_{Y\leq t}] + \bbE_G[{\mathds{1}_{Y\leq t}}^2]\\
                     &= F(t)^2 - 2F(t)G(t) + G(t)\\
                     &= F(t)^2 - 2F(t)G(t) + G(t)^2 - G(t)^2 + G(t)\\
                     &= (F(t)-G(t))^2 + G(t)(1-G(t))
    \end{align*}
\end{proof}

\subsection{Continuous Ranked Probability Score}

For any $F,G\in\calP_1(\bbR)$, the expectation of the CRPS \eqref{eq:CRPS_bs} is :
\begin{align*}
    \bbE_G[\CRPS(F,Y)] &= \bbE_{F,G}|X-Y|-\frac{1}{2}\bbE_F|X-X'|;\\
    &= \int_\bbR (F(z)-G(z))^2 \rmd z + \int_\bbR G(z)(1-G(z))\rmd z,
\end{align*}
where the second term of the last line is the entropy of the CRPS.

\begin{proof}
    \begin{align*}
         \bbE_G[\CRPS(F,Y)] &= \bbE_G\left[\int_\bbR (F(z)-\mathds{1}_{y\leq z})^2\rmd z\right]\\
         &= \int_\bbR \bbE_G\left[(F(z)-\mathds{1}_{y\leq z})^2\right]\rmd z\\
         &= \int_\bbR \bbE_G\left[F(z)^2-2F(z)\mathds{1}_{y\leq z}+\mathds{1}_{y\leq z}^2\right]\rmd z\\
         &= \int_\bbR \left\{F(z)^2-2F(z)\bbE_G\left[\mathds{1}_{y\leq z}\right]+\bbE_G\left[\mathds{1}_{y\leq z}\right]\right\}\rmd z\\
         &= \int_\bbR \left\{F(z)^2-2F(z)G(z)+G(z)\right\}\rmd z\\
         &= \int_\bbR \left\{F(z)^2-2F(z)G(z)+G(z)^2-G(z)^2+G(z)\right\}\rmd z\\
         &= \int_\bbR (F(z)-G(z))^2 \rmd z + \int_\bbR G(z)(1-G(z))\rmd z 
    \end{align*}
\end{proof}

\subsection{Dawid-Sebastiani score}

For any $F,G\in\calP_2(\bbR)$, the expectation of the Dawid-Sebastiani score \eqref{eq:dss_univariate} is :
\begin{equation*}
    \bbE_G[\DSS(F,Y)] = \frac{(\mu_F-\mu_G)^2}{{\sigma_F}^2}+ \frac{{\sigma_G}^2}{{\sigma_F}^2} + 2\log\sigma_F.
\end{equation*}

\begin{proof}
    \begin{align*}
        \bbE_G[\DSS(F,Y)] &= \bbE_G\left[\frac{(Y-\mu_F)^2}{{\sigma_F}^2}+2\log\sigma_F\right]\\
                        &= \frac{\bbE_G\left[(Y-\mu_F)^2\right]}{{\sigma_F}^2}+2\log\sigma_F
    \end{align*}
    Noticing that $\bbE_G\left[(Y-\mu_F)^2\right]=\bbE_G\left[\SE(F,Y)\right]$,
    \begin{align*}
        \bbE_G[\DSS(F,Y)] &= \frac{(\mu_F-\mu_G)^2+{\sigma_G}^2}{{\sigma_F}^2} + 2\log\sigma_F.
    \end{align*}
\end{proof}

\subsection{Error-spread score}

For any $F,G\in\calP_4(\bbR)$, the expectation of the error-spread score \eqref{eq:ESS} is :
\begin{align*}
    \bbE_G[\ESS(F,Y)] &= \left[({\sigma_G}^2-{\sigma_F}^2) + (\mu_G-\mu_F)^2 - \sigma_F \gamma_F (\mu_G-\mu_F)\right]^2\\
    &\ \ \ \ + {\sigma_G}^2\left[2(\mu_G-\mu_F)+(\sigma_G\gamma_G-\sigma_F\gamma_F)\right]^2\\
    & \ \ \ \ +{\sigma_G}^4(\beta_G-{\gamma_G}^2-1),
\end{align*}
where $\mu_F$, $\sigma_F^2$, $\gamma_F$ are the mean, the variance and the skewness of the probabilistic forecast $F$. Similarly, $\mu_G$, $\sigma_G^2$, $\gamma_G$ and $\beta_G$ are the first four centered moments of the distribution $G$. The proof is available in Appendix B of \cite{Christensen2014}.

\subsection{Logarithmic score}

For any $F,G\in\calP(\bbR)$ such that $F$ and $G$ have probability density functions in the class $\calL_1(\bbR)$, the expectation of the logarithmic score \eqref{eq:log_score} is :
\begin{equation*}
    \bbE_G[\mathrm{LogS}(F,Y)] = D_{\mathrm{KL}}(G||F) + \mathrm{H}(F),
\end{equation*}
where $D_{\mathrm{KL}}(G||F)$ is the Kullback-Leibler divergence from $F$ to $G$ and $\mathrm{H}(F)$ is the Shannon entropy of $F$. The proof is straightforward given that the Kullback-Leibler divergence and Shannon entropy are defined as
\begin{align*}
    D_{\mathrm{KL}}(G||F) &= \int_\bbR g(y)\log\left(\frac{g(y)}{f(y)}\right) \rmd y;\\
    \mathrm{H}(F) &= \int_\bbR f(y)\log(f(y))\rmd y.
\end{align*}

\subsection{Hyvärinen score}

For $F,G$ such that their densities $f$ exist, are twice continuously differentiable and satisfy $f'(x)/f(x)\to 0$ as $|x|\to\infty$ and $g'(x)/g(x)\to 0$ as $|x|\to\infty$, the expectation of the Hyvärinen score is :
\begin{align*}
    \bbE_G[\mathrm{HS}(F,Y)] &= \int_\bbR \left( \frac{f'(y)^2}{f(y)^2}-2\frac{f'(y)g'(y)}{f(y)g(y)}\right) g(y) \rmd y\\
    &= \int_\bbR \left( \frac{f'(y)}{f(y)}-\frac{g'(y)}{g(y)}\right)^2 g(y) \rmd y - \int_\bbR \frac{g'(y)^2}{g(y)^2} g(y) \rmd y
\end{align*}
where the last formula shows the entropy of the Hyvärinen score (second term on the right-hand side).

\begin{proof}
    \begin{align*}
        \bbE_G[\mathrm{HS}(F,Y)] &= \bbE\left[2\frac{f''(Y)}{f(Y)}-\frac{f'(Y)^2}{f(y)^2}\right]\\
        &= 2\int_\bbR \frac{f''(y)}{f(y)}g(y)\rmd y - \int_\bbR \frac{f'(y)^2}{f(y)^2} g(y)\rmd y
    \end{align*}
    Integrating by part the integral of the first term on the right-hand side leads to :
    \begin{align*}
        \int_\bbR \frac{f''(y)}{f(y)}g(y)\rmd y &= \left[\frac{f'(y)}{f(y)}g(y)\right]^{+\infty}_{-\infty} - \int_\bbR f'(y)\frac{g'(y)f(y)-g(y)f'(y)}{f(y)^2}\rmd y\\
        &= - \int_\bbR \frac{f'(y)g'(y)}{f(y)g(y)}g(y)\rmd y + \int_\bbR \frac{f'(y)^2}{f(y)^2}g(y)\rmd y\\
    \end{align*}
    The boundary term is null since $f'(x)/f(x)\to 0$ as $|x|\to\infty$ and $g$ is a probability density function.\\
    Thus,
    \begin{align*}
        \bbE_G[\mathrm{HS}(F,Y)] &= - 2\int_\bbR \frac{f'(y)g'(y)}{f(y)g(y)}g(y)\rmd y + 2\int_\bbR \frac{f'(y)^2}{f(y)^2}g(y)\rmd y - \int_\bbR \frac{f'(y)^2}{f(y)^2} g(y)\rmd y\\
        &= - 2\int_\bbR \frac{f'(y)g'(y)}{f(y)g(y)}g(y)\rmd y + \int_\bbR \frac{f'(y)^2}{f(y)^2}g(y)\rmd y\\
        &= \int_\bbR \left(\frac{f'(y)^2}{f(y)^2}-2\frac{f'(y)g'(y)}{f(y)g(y)}\right)g(y)\rmd y
    \end{align*}
\end{proof}

\subsection{Quadratic score}

For any $F,G\in\calL_2(\bbR)$, the expectation of the quadratic score is :
\[
    \bbE_G[\mathrm{QuadS}(F,Y)] = \lVert f\rVert_2^2 - 2 \langle f,g\rangle,
\]
where $\langle f,g\rangle=\int_\bbR f(y)g(y)\rmd y$.\\

\subsection{Pseudospherical score}

For any $F,G\in\calL_\alpha(\bbR)$, the expectation of the quadratic score is :
\[
    \bbE_G[\mathrm{PseudoS}(F,Y)] = -\frac{\langle f^{\alpha-1},g\rangle}{\lVert f\rVert_\alpha^{\alpha-1}},
\]
where $\langle f^{\alpha-1},g\rangle=\int_\bbR f(y)^{\alpha-1}g(y)\rmd y$.\\

\section{Expected multivariate scoring rules}\label{appendix:multi_expected_scores}

\subsection{Squared error}

For any $F,G\in\calP_2(\bbR^d)$, the expectation of the squared error \eqref{eq:ES_multivariate} is :
\begin{equation*}
    \bbE_G[\SE(F,\bmY)] = \lVert \bm{\mu}_F-\bm{\mu}_G\rVert^2_2 + \mathrm{tr}(\Sigma_G),
\end{equation*}
where $\bm{\mu}_F$ is the mean vector of the distribution $F$ and $\bm{\mu}_G$ and ${\Sigma_G}^2$ are the mean vector and the covariance matrix of the distribution $G$.
\begin{proof}
    Let $T_i$ denote the projection on the $i$-th margin.
    \begin{align*}
        \bbE_G[\SE(F,\bmY)] &= \bbE_G[\lVert \bm{\mu}_F-\bmY\rVert^2_2]\\
        &= \bbE_G\left[\sum_{i=1}^d (\bm{\mu}_{T_i(F)}-T_i(\bmY))^2\right]\\
        &= \sum_{i=1}^d\bbE_{T_i(G)}\left[\SE(T_i(F),Y)\right]\\
        &= \sum_{i=1}^d \left( (\mu_{T_i(F)}-\mu_{T_i(G)})^2+\sigma_{T_i(G)}^2\right)\\
        &= \lVert \bm{\mu}_F-\bm{\mu}_G\rVert^2_2 + \mathrm{tr}(\Sigma_G)\\
    \end{align*}            
\end{proof}

\subsection{Dawid-Sebastiani score}

For any $F,G\in\calP_2(\bbR^d)$, the expectation of the Dawid-Sebastiani score is :
\[
    \bbE_G[\DSS(F,\bmY)] = \log(\det \Sigma_F) + (\bm{\mu}_F-\bm{\mu}_G)^T\Sigma_F^{-1}(\bm{\mu}_F-\bm{\mu}_G) + \mathrm{tr}(\Sigma_G\Sigma_F^{-1}).
\]
The proof is available in the original article \citep{Dawid1999}.

\subsection{Energy score}

In a general setting, the expected energy score does not simplify. For any $F,G\in\calP_\beta(\bbR^d)$, the expected energy score \eqref{eq:ES} is :
\[
    \bbE_G[\ES_\beta(F,\bmY)] = \bbE_{F,G}\lVert \bmX-\bmY\lVert^\beta_2-\frac{1}{2}\bbE_F\lVert \bmX-\bmX'\lVert^\beta_2.
\]

\subsection{Variogram score}

For any $F,G\in\calP(\bbR^d)$ such that the $2p$-th moments of all their univariate margins are finite, the expected variogram score of order $p$ \eqref{eq:VS} is :
\begin{align*}
    \bbE_G[\VS_p(F,\bmY)] = \sum_{i,j=1}^{d} w_{ij}\left(\bbE_F\left[|X_i-X_j|^p\right]^2-2\bbE_F\left[|X_i-X_j|^p\right]\bbE_G[|Y_i-Y_j|^p]+\bbE_G[|Y_i-Y_j|^{2p}]\right).
\end{align*}

\begin{proof}
    \begin{align*}
        \bbE_G[\VS_p(F,\bmY)] 
        &=\bbE_G\left[\sum_{i,j=1}^{d} w_{ij}\left(\bbE_F\left[|X_i-X_j|^p\right]-|Y_i-Y_j|^p\right)^2\right]\\
        &=\bbE_G\left[\sum_{i,j=1}^{d} w_{ij}\left(\bbE_F\left[|X_i-X_j|^p\right]^2-2\bbE_F\left[|X_i-X_j|^p\right] |Y_i-Y_j|^p+|Y_i-Y_j|^{2p}\right)\right]\\
        &=\sum_{i,j=1}^{d} w_{ij}\left(\bbE_F\left[|X_i-X_j|^p\right]^2-2\bbE_F\left[|X_i-X_j|^p\right]\bbE_G[|Y_i-Y_j|^p]+\bbE_G[|Y_i-Y_j|^{2p}]\right).
    \end{align*}
\end{proof}

\subsection{Logarithmic score}

For any $F,G\in\calP(\bbR^d)$ such that $F$ and $G$ have probability density functions that belong to $\calL_1(\bbR^d)$, the expectation of the logarithmic score is analogous to its univariate version :
\begin{equation*}
    \bbE_G[\mathrm{LogS}(F,\bmY)] = D_{\mathrm{KL}}(G||F) + \mathrm{H}(F),
\end{equation*}
where $D_{\mathrm{KL}}(G||F)$ is the Kullback-Leibler divergence from $F$ to $G$ and $\mathrm{H}(F)$ is the Shannon entropy of $F$.

\begin{align*}
    D_{\mathrm{KL}}(G||F) &= \int_{\bbR^d} g(\bmy)\log\left(\frac{g(\bmy)}{f(\bmy)}\right) \rmd \bmy\\
    \mathrm{H}(F) &= \int_{\bbR^d} f(\bmy)\log(f(\bmy))\rmd \bmy.
\end{align*}

\subsection{Hyvärinen score}

For $F,G\in\calP(\bbR^d)$ such that their probability density functions $f$ and $g$ such that they are twice continuously differentiable and satisfying $\nabla f(x)\to 0$ and $\nabla g(x)\to 0$ as $\lVert x\rVert\to\infty$, the expectation of the Hyvärinen score is :
\[
    \bbE[\mathrm{HS}(F,\bmY)] = \int_{\bbR^d} g(y) \langle\nabla\log(f(y))-2\nabla\log(g(y)),\nabla\log(f(y))\rangle g(y) \rmd y
\]
where $\nabla$ is the gradient operator and $\langle\cdot,\cdot\rangle$ is the scalar product. The proof is similar to the proof for the univariate case using integration by parts and Stoke's theorem \citep{Parry2012}.

\subsection{Quadratic score}

For any $F,G\in\calL_2(\bbR^d)$, the expectation of the quadratic score is analogous to its univariate version :
\[
    \bbE_G[\mathrm{QuadS}(F,\bmY)] = \lVert f\rVert_2^2 - 2 \langle f,g\rangle,
\]
where $\langle f,g\rangle=\int_{\bbR^d} f(\bmy)g(\bmy)\rmd \bmy$.\\

\subsection{Pseudospherical score}

For any $F,G\in\calL_\alpha(\bbR^d)$, the expectation of the quadratic score is analogous to its univariate version :
\[
    \bbE_G[\mathrm{PseudoS}(F,\bmY)] = -\frac{\langle f^{\alpha-1},g\rangle}{\lVert f\rVert_\alpha^{\alpha-1}},
\]
where $\langle f^{\alpha-1},g\rangle=\int_{\bbR^d} f(\bmy)^{\alpha-1}g(\bmy)\rmd \bmy$.\\

\section{Proofs}\label{appendix:proofs}

\subsection{Proposition~\ref{prop:transformation_SR}}\label{appendix:aggregation-transformation}

\begin{proof}[Proof of Proposition~\ref{prop:transformation_SR}]
    Let $\calF\subset\calP(\bbR^d)$ be a class of Borel probability measure on $\bbR^d$ and let $F\in\calF$ be a forecast and $y\in\bbR^d$ an observation. Let $T:\bbR^d\to\bbR^k$ be a transformation and let $\rmS$ be a scoring rule on $\bbR^k$ that is proper relative to $T(\calF)=\{\calL(T(\bmX)), X\sim F\in\calF\}$.
    \begin{align*}
        \bbE_G\left[\rmS_T(F,\bmY)\right] &= \bbE_G\left[\rmS(T(F)),T(\bmY))\right]\\
                         &= \bbE_{T(G)}\left[\rmS(T(F),\bmY)\right]
    \end{align*}
    Given that $T(F),T(G)\in T(\calF)$ and $\rmS$ is proper relative to $T(\calF)$, 
    \begin{align}
        &\bbE_{T(G)}\left[\rmS(T(G),\bmY)\right] \leq \bbE_{T(G)}\left[\rmS(T(F),\bmY)\right]\nonumber\\ 
        \Leftrightarrow\ \ &\bbE_G\left[\rmS_T(G,\bmY)\right] \leq \bbE_G\left[\rmS_T(F,\bm\bmY)\right]\label{eq:proof_propriety}
    \end{align}
\end{proof}

\begin{proof}[Proof of the strict propriety case in Proposition~\ref{prop:transformation_SR}]
    The notations are the same as the proof above except the following. Let $T:\bbR^d\to\bbR^k$ be an injective transformation and let $\rmS$ be a scoring rule on $\bbR^k$ that is strictly proper relative to $T(\calF)=\{\calL(T(\bmX)), X\sim F\in\calF\}$.\\

    The equality in Equation~\eqref{eq:proof_propriety} leads to :
    \begin{align*}
        &\bbE_G\left[\rmS_T(G,\bmY)\right] = \bbE_G\left[\rmS_T(F,\bmY)\right]\\
        \Leftrightarrow\ \ &\bbE_{G}\left[\rmS(T(G),T(\bmY))\right] = \bbE_{G}\left[\rmS(T(F),T(\bmY))\right]\\
        \Leftrightarrow\ \ &\bbE_{T(G)}\left[\rmS(T(G),\bmY)\right] = \bbE_{T(G)}\left[\rmS(T(F),\bmY)\right]\\
    \end{align*}
    The fact that $\rmS$ is strictly proper relative to $T(\calF)$ leads to $T(F)=T(G)$, and finally since $T$ is injective, we have $F=G$.
\end{proof}

\subsection{Proposition~\ref{prop:series-representation}}\label{appendix:proof-kernel}

\begin{proof}[Proof of Proposition~\ref{prop:series-representation}]
The proof relies on the reproducing kernel Hilbert space (RKHS) representation of the kernel scoring rule $\rmS_\rho$. For a background on kernel scoring rule, maximum mean discrepancies and RKHS, we refer to \cite{Smola2007} or \citet[Section 4]{steinwart2008}. 

Let $\mathcal{H}_\rho$ denote the 
RKHS  associated with $\rho$. We recall that  $\mathcal{H}_\rho$ contains all the functions $\rho(\bmx,\cdot)$ and that the inner product on $\mathcal{H}_\rho$ satisfies the property
\[
    \langle \rho(\bmx_1,\cdot),\rho(\bmx_2,\cdot) \rangle_{\mathcal{H}_\rho}=\rho(\bmx_1,\bmx_2).
\]  
The \textit{kernel mean embedding} is a linear application $\Psi_\rho:\mathcal{P}_\rho\to\mathcal{H}_\rho$ mapping an admissible distribution $F\in\mathcal{P}_\rho$ into a function $\Psi_\rho(F)$ in the RKHS and such that the image of the point measure $\delta_{\bmx}$ is $\rho(\bmx,\cdot)$.  Equation~\eqref{eq:def-S_k-ensemble} giving the kernel scoring rule for an ensemble prediction $F=\frac{1}{M}\sum_{m=1}^M \delta_{\bmx_m}$ can be written as
\begin{align*}
    \rmS_\rho(F,\bmy) &= \frac{1}{2}\langle \Psi_\rho(F)-\Psi_\rho(\delta_{\bmy}),\Psi_\rho(F)-\Psi_\rho(\delta_{\bmy})\rangle _{\mathcal{H}_\rho}\\
    &= \frac{1}{2}\|\Psi_\rho(F-\delta_{\bmy}) \|_{\mathcal{H}_\rho}^2.
\end{align*}
The properties of the kernel mean embedding ensure that this relation still holds for all $F\in\mathcal{P}_\rho$. As a consequence, if $(T_l)_{l\geq 1}$ is an Hilbertian basis of $\mathcal{H}_\rho$, we have
\begin{align*}
    \rmS_\rho(F,y) &= \frac{1}{2}\|\Psi_\rho(F-\delta_{\bmy}) \|_{\mathcal{H}_\rho}^2\\
    &= \frac{1}{2}\sum_{l\geq 1}\langle \Psi_\rho(F-\delta_{\bmy}),T_l\rangle_{\mathcal{H}_\rho}^2.
\end{align*}
Finally, the properties of the kernel mean embedding ensure that, for all $T\in\mathcal{H}_\rho$,
\[
    \langle\Psi_\rho(F-\delta_{\bmy}),T\rangle_{\mathcal{H}_\rho}=\int_{\bbR^d}T(\bmx)(F-\delta_{\bmy})(\rmd \bmx)=\bbE_F[T(\bmX)]-T(\bmy)
\]
whence the result follows.
\end{proof}

\subsection{Proof of examples illustrating Proposition~\ref{prop:series-representation}}\label{appendix:proof-kernel-examples}

Next, we illustrate the Proposition~\ref{prop:series-representation} and provide some computations in two cases: the Gaussian kernel scoring rule and the continuous rank probability score (CRPS).\\

\noindent\textbf{Gaussian Kernel Scoring Rule}. This is the scoring rule related to the Gaussian kernel 
\[
    \rho(x_1,x_2)=\exp(-(x_1-x_2)^2/2),\quad x_1,x_2\in\bbR.
\]
Using a series expansion of the exponential function, we have
\[
    \rho(x_1,x_2)=\mathrm{e}^{-x_1^2/2}\mathrm{e}^{-x_2^2/2}\sum_{l\geq 0}\frac{(x_1x_2)^l}{l!}=\sum_{l\geq 0}T_l(x_1)T_l(x_2)
\]
with $T_l$ the transformation defined, for $l\geq 0$, by
\[
    T_l(x)=\frac{1}{\sqrt{l!}}\mathrm{e}^{-x^2/2}x^l.
\]
As a consequence, the Gaussian kernel scoring rule writes, for all $F\in\mathcal{P}(\bbR)$ and $y\in\bbR$,
\begin{align*}
    \rmS_\rho(F,y)&=\frac{1}{2}\int_{\bbR\times\bbR}\rho(x_1,x_2)(F-\delta_y)(\rmd x_1)(F-\delta_y)(\rmd x_2)\\
    &= \frac{1}{2}\int_{\bbR\times\bbR}\Big(\sum_{l\geq 0}T_l(x_1)T_l(x_2)\Big)(F-\delta_y)(\rmd x_1)(F-\delta_y)(\rmd x_2)\\
    &= \frac{1}{2}\sum_{l\geq 0}\Big(\int_{\bbR}T_l(x)(F-\delta_y)(\rmd x)\Big)^2\\
    &= \frac{1}{2}\sum_{l\geq 0} \Big(\bbE_F[T_l(X)]-T_l(y)\Big)^2.
\end{align*}\\

\noindent\textbf{Continuous Ranked Probability Score}. The CRPS is the scoring rule with kernel $\rho(x_1,x_2)=|x_1|+|x_2|-|x_1-x_2|$. This kernel is the covariance of the Brownian motion on $\bbR$ and its RKHS is known to be the Sobolev space $\mathrm{H}^1=\mathrm{H}^1(\bbR)$, see \cite{Berlinet2004}. We recall the definition of the Sobolev space
\[
    \mathrm{H}^1=\left\{f\in \mathcal{C}(\bbR,\bbR) \colon f(0)=0, \dot f\in L^2(\bbR)\right\},
\]
where $\dot f$ denotes the derivative of $f$ assumed to be defined almost everywhere and square-integrable. The inner product on $\mathrm{H}^1$ is defined by
\[
    \langle f_1,f_2\rangle_{H^1}=\int_{\bbR}\dot f_1(x)\dot f_2(x)\rmd x
\]
and one can  easily check the fundamental  relation
\[
    \langle \rho(x_1,\cdot),\rho(x_2,\cdot)\rangle_{H^1}=\int_{\bbR} \dot \rho(x_1,x)\dot \rho(x_2,x)\rmd x = \rho(x_1,x_2).
\]
Here the derivative $\dot \rho(x_1,x)=\mathds{1}_{[0,x_1]}(x)$ is taken with respect to the second variable $x$. Then, we consider the Haar system  defined as the collection of functions 
\[
    H^0_{l}(x)=H^0(x-l) \quad \mbox{and}\quad  H^1_{l,m}(x)=2^{m/2}H^1(2^{m}x-l) ,\quad l\in\bbZ,\, m\geq 0,  
\]
with $H^0(x)=\mathds{1}_{[0,1)}(x)$ and $H^1(x)=\mathds{1}_{[0,1/2)}(x)-\mathds{1}_{[1/2,1)}(x)$. Since the Haar system is an orthonormal basis of the space $L^2(\bbR)$ and the map $f\in H^1\mapsto \dot f\in L^2$ is an isomorphism between Hilbert spaces, we obtain an orthonormal basis of $\mathrm{H}^1(\bbR)$ by considering the primitives vanishing at $0$ of the Haar basis functions. Setting $T^0(x)=x\mathds{1}_{[0,1)}(x)+\mathds{1}_{[1,+\infty)}(x)$ and $T^1(x)=\big(1/2-|x-1/2|\big)\mathds{1}_{[0,1]}(x)$ the primitive functions of $H^0$ and $H^1$ respectively, we obtain the system 
\[
    T^0_{l}(x)= T^0(x-l),\quad T^1_{l,m}(x)=2^{-m/2}T^1(2^{m}x-l),\quad l\in\bbZ,\, m\geq 0.
\]
The series representation of the CRPS is then deduced from Proposition~\ref{prop:series-representation} and its proof since the collection $\{T_{l,m}\colon l\in\bbZ, m\geq 0\}$, is an orthonormal basis of the RKHS  associated with the kernel $\rho$ of the CRPS.

\section{General form of Corollary~\ref{cor:agg+transform_SR}}\label{appendix:general_cor}

\begin{corollary}
    Let $\calT=\{T_i\}_{{1\leq i\leq m}}$ be a set of transformations from $\bbR^d$ to $\bbR^k$. Let $\calS=\{\rmS_i\}_{{1\leq i\leq m}}$ be a set of proper scoring rules such that $\rmS_i$ is proper relative to $T_i(\calF)$, for all $1\leq i\leq m$. Let $\bm{w}=\{w_i\}_{{1\leq i\leq m}}$ be nonnegative weights. Then the scoring rule
    \begin{equation*}
        \rmS_{\calS_\calT,\bm{w}}(F,\bmy) = \sum_{i=1}^m w_i {\rmS_i}_{T_i}(F,\bmy) = \sum_{i=1}^m w_i {\rmS_i}(T_i(F),T_i(\bmy))
    \end{equation*}
    is proper relative to $\calF$.
\end{corollary}

\section{Scoring rules of the simulation study}\label{appendix:sr_simu}

The following formulas are deduced for a probabilistic forecast $F$ taking the form of the Gaussian random field model of Equation~\eqref{eq:obs_simu}. The formulas of the aggregated univariate scoring rules can be obtained from the formulas in \cite{Gneiting2007} and \cite{Jordan2019} and, thus, are not presented here. We focus on the expression of the variogram score and the CRPS of spatial mean.

\subsection*{Variogram Score}

\begin{align*}
    \VS_p(F,\bmy) &= \sum_{s,s'\in\calD}  w_{ss'} \left(\bbE_{F}[|X_s-X_{s'}|^p] -|y_s-y_{s'}|^p\right)^2
\end{align*}
For $X\sim\calN(\mu,\sigma^2)$, the absolute moment is \citep{Winkelbauer2014} :
\begin{equation}\label{eq:absolute_moment}
    \bbE[|X|^\nu] = \sigma^\nu 2^{\nu/2}\frac{\Gamma\left(\frac{\nu+1}{2}\right)}{\sqrt{\pi}}{}_1F_1\left(-\nu/2,1/2;-\frac{\mu^2}{2\sigma^2}\right),
\end{equation}
where ${}_1F_1$ is the confluent hypergeometric function of the first kind. For $X\sim F$, 
\begin{align*}
    X_s-X_{s'} &\sim \calN(\mu_s-\mu_{s'},{\sigma_s}^2+{\sigma_{s'}}^2-2\mathrm{cov}(F_{s},F_{s'})\\
    &\sim\calN(0,2\sigma^2(1-e^{-\left(\frac{\lVert s-s'\rVert}{\lambda}\right)^{\beta}})).
\end{align*}

This leads to
\begin{align*}
    \bbE_G[|X_s-X_{s'}|^p] &= \left(2\sigma^2(1-e^{-\left(\frac{\lVert s-s'\rVert}{\lambda}\right)^{\beta}})\right)^{p/2} 2^{p/2}\frac{\Gamma\left(\frac{p+1}{2}\right)}{\sqrt{\pi}}{}_1F_1\left(-p/2,1/2;-\frac{(\mu_s-\mu_{s'})^2}{4\sigma^2(1-e^{-\left(\frac{\lVert s-s'\rVert}{\lambda}\right)^{\beta}})}\right)\\
    &= 2^p \sigma^p \left(1-e^{-\left(\frac{\lVert s-s'\rVert}{\lambda}\right)^{\beta}}\right)^{p/2} \frac{\Gamma\left(\frac{p+1}{2}\right)}{\sqrt{\pi}}{}_1F_1\left(-p/2,1/2;0\right)\\
    &= 2^p \sigma^p \left(1-e^{-\left(\frac{\lVert s-s'\rVert}{\lambda}\right)^{\beta}}\right)^{p/2} \frac{\Gamma\left(\frac{p+1}{2}\right)}{\sqrt{\pi}}
\end{align*}

Finally,
\begin{align*}
    \VS_p(F,\bmy) &= \sum_{s,s'\in\calD}  w_{ij} \left(\bbE_{G}[|X_s-X_{s'}|^p] -|y_s-y_{s'}|^p\right)^2\\
    &= \sum_{s,s'\in\calD}  w_{ij} \Biggl(\left(2\sigma^2(1-e^{-\left(\frac{\lVert s-s'\rVert}{\lambda}\right)^{\beta}})\right)^{p/2} 2^{p/2}\frac{\Gamma\left(\frac{p+1}{2}\right)}{\sqrt{\pi}} - |y_s-y_{s'}|^p \Biggl)^2
\end{align*}

\subsection*{p-Variation Score}

\begin{align*}
    \mathrm{pVS}(F,\bmy) &= \sum_{\bms\in\calD^\ast} w_{\bms} \SE_{T_{p-var,\bms}}(F,\bmy);\\
    &= \sum_{\bms\in\calD^\ast} w_{\bms} (\bbE_F[T_{p-var,\bms}(\bmX)]-T_{p-var,\bms}(\bmy))^2,
\end{align*}

Denote $Z =  \bmX_{\bms+(1,1)}-\bmX_{\bms+(1,0)}-\bmX_{\bms+(0,1)}+\bmX_{\bms}$. For $X\sim F$, we have $Z\sim\calN(\mu_Z,\sigma_Z^2)$ with
\[
    \mu_Z = \mu_{\bms+(1,1)}-\mu_{\bms+(1,0)}-\mu_{\bms+(0,1)}+\mu_{\bms} = 0
\]
and 
\begin{align*}
    \sigma_Z^2 &= \sigma_{\bms+(1,1)}^2+\sigma_{\bms+(1,0)}^2+\sigma_{\bms+(0,1)}^2+\sigma_{\bms}^2\\
     &\ \ \ \ -2\mathrm{cov}(F(\bms+(1,1)),F(\bms+(1,0))) - 2\mathrm{cov}(F(\bms+(1,1)),F(\bms+(0,1)) + 2\mathrm{cov}(F(\bms+(1,1)),F(\bms))\\
     &\ \ \ \ +2\mathrm{cov}(F(\bms+(1,0)),F(\bms+(0,1))) - 2\mathrm{cov}(F(\bms+(1,0)),F(\bms))\\
     &\ \ \ \ -2\mathrm{cov}(F(\bms+(0,1)),F(\bms))\\
     &= 4\sigma^2 (1+e^{-(\sqrt{2}/\lambda)^\beta}-2e^{-(1/\lambda)^\beta})
\end{align*}

Using \eqref{eq:absolute_moment}, this leads to
\begin{align*}
    \bbE_F[T_{p-var,\bms}(\bmX)] &= \left(4\sigma^2 (1+e^{-(\sqrt{2}/\lambda)^\beta}-2e^{-(1/\lambda)^\beta})\right)^{p/2} 2^{p/2}\frac{\Gamma\left(\frac{p+1}{2}\right)}{\sqrt{\pi}}{}_1F_1\left(-p/2,1/2;0\right)\\
    &= \left(4\sigma^2 (1+e^{-(\sqrt{2}/\lambda)^\beta}-2e^{-(1/\lambda)^\beta})\right)^{p/2} 2^{p/2}\frac{\Gamma\left(\frac{p+1}{2}\right)}{\sqrt{\pi}}
\end{align*}

Finally,
\begin{align*}
    \mathrm{pVS}(F,\bmy) &= \sum_{\bms\in\calD^\ast} w_{\bms} \SE_{T_{p-var,\bms}}(F,\bmy)\\
    &= \sum_{\bms\in\calD^\ast}  w_{\bms} \Biggl(\left(4\sigma^2 (1+e^{-(\sqrt{2}/\lambda)^\beta}-2e^{-(1/\lambda)^\beta})\right)^{p/2} 2^{p/2}\frac{\Gamma\left(\frac{p+1}{2}\right)}{\sqrt{\pi}} - |y_{\bms+(1,1)}-y_{\bms+(1,0)}-y_{\bms+(0,1)}+y_{\bms}|^p \Biggl)^2
\end{align*}

\subsection*{CRPS of spatial mean}

The CRPS of spatial mean is defined as 
\begin{align*}
    \CRPS_{\mathrm{mean}_\calP,\bm{w_\calP}}(F,\bmy) &= \sum_{P\in\calP} w_P \CRPS_{\mathrm{mean}_P}(F,\bmy)\\
    &= \sum_{P\in\calP} w_P \CRPS(\mathrm{mean}_P(F),\mathrm{mean}_P(\bmy)),
\end{align*}
where $\calP$ is an ensemble of spatial patches and $w_P$ is the weight associated with a patch $P\in\calP$. The mean of Gaussian marginals follows a Gaussian distribution : 
\[
    \mathrm{mean}_P(F) \sim \calN(\sum_{s\in P} \mu_{s},\frac{\sigma^2}{|P|^2}\sum_{s,s'\in P} e^{-(\frac{\lVert s-s'\rVert}{\lambda})^{\beta}}) = \calN(\mu_P,\sigma_P^2),
\]
where $|P|$ is the cardinal of the patch $P$ (i.e., the number of grid points belonging to $P$).\\

Finally,
\begin{align*}
    \CRPS_{\mathrm{mean}_\calP,\bm{w_\calP}}(F,\bmy) = \sum_{P\in\calP} w_P \CRPS(\calN(\mu_P,\sigma_P^2),\mathrm{mean}_P(\bmy)).
\end{align*}

\end{document}